\documentclass[a4paper, 12pt]{article}
\usepackage[onehalfspacing]{setspace}
\usepackage{amssymb, soul}
\usepackage[utf8]{inputenc}
\usepackage{amsmath}
\usepackage{amsthm}
\usepackage{textcomp}
\usepackage{graphicx}
\usepackage{tikz}
\usepackage{ circuitikz }
\usetikzlibrary{positioning}
\usetikzlibrary{calc,intersections}
\usepackage{float}
\usepackage{bbm}
\usepackage{pgf}
\usepackage{natbib}
\usepackage{hyperref}
\usepackage{xcolor}
\usepackage{accents}
\usepackage{subfigure}
\usepackage{appendix}
\usepackage{enumitem}
\usetikzlibrary{decorations.pathreplacing,calligraphy}
\usetikzlibrary{decorations.text,calc,arrows.meta}
\usetikzlibrary{arrows,automata}
\usepackage{geometry}
\newtheorem{theorem}{Theorem}

\newtheorem{assumption}{Assumption}

\newtheorem{case*}{Case}
\newtheorem{claim}{Claim}

\newtheorem{corollary}{Corollary}

\newtheorem{definition}{Definition}

\newtheorem{lemma}{Lemma}

\newtheorem{proposition}{Proposition}

\title{Sequential Search with Planning\thanks{We are grateful to Deepak Kumar, Jason Lepore, Kevin Reffett, and Santanu Roy for their helpful suggestions at various stages of this work. The authors are listed in random order, following the proposal of \cite{rr18}.}}

\author{
Ruhi Sonal$^{\text{r}}$\thanks{Department of Social Sciences and Humanities, IIIT Delhi.}
\and Saptarshi Mukherjee$^{\text{r}}$\thanks{Department of Humanities and Social Sciences, IIT Delhi.}
\and Abhinaba Lahiri$^{\text{r}}$\thanks{Centre for Mathematical and Computational Economics, IIT Jodhpur.}
\and Aniruddha Ghosh$^{\text{r}}$\thanks{Orfalea College of Business, Cal Poly San Luis Obispo.}
}
\begin{document}
\maketitle                  
\begin{abstract}

\noindent Sequential development of a new product or technology, or natural resource exploration, often progresses through ordered stages with uncertain rewards and requires costly (ex-ante) planning to make future stages accessible. We model this process as an ordered Pandora's box problem where a decision-maker first chooses an initial scope, paying a cost that rises with the number of stages made accessible, and may later expand the scope at a marginal adjustment cost. Since the paid planning costs are sunk, the continuation values depend on the state variable “paid scope". We prove existence and uniqueness of scope-dependent reservation values, characterize the optimal search strategy as a threshold rule indexed by paid scope, and derive comparative statics. Interactions among three economic forces shape the optimal behavior -- a {\it{guarantee effect}} (a higher current best offer reduces the expected improvement from the next stage and induces earlier stopping), a {\it{paid-scope effect}} (a larger prepaid scope lowers the marginal cost of future access, raises the continuation value, and supports continuation at higher guarantees), and a {\it{remaining-horizon effect}} (fewer stages remaining shrink the option value of continuing).  Two examples illustrate how these forces generate distinct planning and search patterns under normal and fat-tailed rewards.

\medskip

\noindent{JEL classification:} D25, D81, D82, D83 

\medskip
\noindent{Keywords:}  Pandora's box, R\&D planning, search, information acquisition

	\end{abstract}
	
	\newpage

\tableofcontents
\newpage

\section{Introduction}

A new product or technology is often developed over multiple stages of
R\&D, where outcomes at each stage are uncertain. This uncertainty renders
the development process iterative \citep{eisenhardt} and often
time-consuming. Firms therefore plan their R\&D activities in advance for
several future periods, incurring ex-ante costs such as planning
expenditures and initial outlays required to initiate the program. A large
literature documents that the adoption and development of new technologies
is slow and incremental \citep{FARZIN1998779}. In such settings, rewards are inherently sequential as accessing later
stages of development typically requires completing earlier ones. For
example, the development of AI and large language models involves
successive stages of design, training, and evaluation, each building on
prior investments.\footnote{The decision of when to release a product
corresponds to a stopping decision wherein the firm may commercialize at an
intermediate stage, with ex ante uncertain rewards.} This raises a natural
economic question: how far ahead should a firm plan its exploration when
the outcomes are uncertain and access to future stages is costly? 

Further, exploration and search strategies are often adaptive. There are many examples in which an initial plan to develop a final product or technology was abandoned, or postponed for a significant period, during intermediate phases because an early-stage innovation had already generated substantial returns. For instance, 3M initially set out to develop a strong adhesive but instead monetized an unexpected intermediate-stage product: a low-adhesion prototype marketed as Press-n-Peel in 1977, which later became popular as Post-it Notes. Similarly, in the case of Dyson vacuum cleaners, the basic dual-cyclone prototype was already sufficiently promising that later, more ambitious stages were frozen. Conversely, product-development plans may be abandoned when the expected rewards no longer appear as attractive as the initial plan had suggested.

In an influential paper, \citet{Roberts1981-yd} study the sequential development of a project in which a firm invests funds in stages, observes interim outcomes and decides
whether to continue or abandon. In this setting, the firm learns progressively about a single terminal payoff. Their central result is a threshold rule that states that  continuation is optimal when the expected value of information from the
next stage exceeds its cost. In their framework, advancing to the next
stage always remains a free option. This is \textit{not} realistic in several economic settings, where firms must {\it{commit}} resources in advance to make future stages feasible -- consider
reserving infrastructure, securing regulatory approvals, or assembling
teams. These commitments determine the set of stages that can be accessed, if required. Expanding this set ex post is possible,
but typically more costly at the margin than committing to a broader plan
ex ante.\footnote{For instance, firms investing in large-scale AI systems
must secure data center capacity, training data, and safety infrastructure
prior to model development. Similarly, clinical trials require advance
commitments to recruitment and manufacturing. In both cases, expanding
scope mid-course is feasible but costly.} The dynamic decision-making in such cases requires forming a concrete plan for a horizon and investing accordingly, while leaving the option for altering the plan open, thereby making the decision process adaptive. 



We study a sequential search problem with costly planning where at each stage an option (interchangeably, alternative) offers an uncertain reward and where the sequence of the options is exogenously given. As evident and well established in the literature, exploring each option is costly and hence the optimal search is a natural problem. Without any planning costs, this problem is a variation of the Pandora's box problem in \citet{weitzman} with the restriction of the exogenously given sequence of the alternatives. This approach is similar to that in \citet{boodaghians}, which deals with the Pandora's box
problem of \citet{weitzman} when the options appear in the form of a \textit{tree}.\footnote{Options are said to appear in the form of a tree when the boxes are nodes of a tree, and can only be opened after their parents are opened. See \cite{boodaghians} for more details.} However, there is an important additional structure in our model. The decision-maker chooses an
initial {\it{planning horizon}}---referred to as \emph{scope}. Choosing a scope is costly and this cost is
increasing in its horizon. This scope cost reflects the initial investment for the projected search. As the DM searches, they have the option to alter the scope at an {\it{additional}} adjustment cost. Thus  the search is {\it{adaptive}}.\footnote{When scope and adjustment costs vanish, the model reduces
to the ordered-search benchmark of \citet{boodaghians}.}

This changes the optimal search rule in \textit{three} important ways.  First, the standard solution concept in optimal search problems is the {\it{reservation value}}, an amount that makes the decision maker indifferent to continuing the search. In our model, reservation values become
\emph{scope-dependent}. That is, at a given stage and realized history, two
decision-makers who differ only in their paid scope may optimally make
different continuation decisions. Second, planning and search decisions interact because current scope determines the cost of future access. A larger prepaid scope makes later stages cheaper to reach and therefore raises the \textit{continuation value}, that is the payoff that continuing the search is expected to generate. A smaller scope leaves more future access costs unpaid, which lowers the continuation value. This can induce the decision-maker to stop earlier. Third, the problem is no longer indexable in the usual sense -- a
single reservation value per stage is insufficient to characterize the optimal
behavior.\footnote{In standard formulations of the Pandora's box problem,
each box can be assigned a single reservation value that determines its
role in the optimal policy. In our setting, such a reduction is not
possible because the value of continuing depends on the endogenous cost of
accessing future stages, which in turn depends on past planning
decisions.}

\subsection{Optimal search strategy: scope-dependent reservation values}
We characterize the optimal search strategy in this setting and show that it takes the form of a
threshold rule {\it{indexed}} by the paid scope. More specifically, the decision-maker evaluates all possible scopes (or planning horizons) and selects the optimal scope to begin with. Corresponding to each such scope, there is a reservation value attached to each stage. The optimal strategy compares the default value (to begin with) and the  reservation value for the optimal scope to decide if to continue the search. The rule is adaptive and hence is fully flexible in terms of stopping decision as well as revising the scope at any stage. Our methodology of deriving the optimal strategy follows standard dynamic programming techniques. However,  our idea of adaptive scope in such search problems, as motivated through several examples, is modeled through the scope-dependent reservation values. This provides a parsimonious approach to model adaptive search behavior normally observed through the examples cited earlier, especially in 
R\&D activities or natural resource extractions. 

The scope-dependent reservation values that emerge from costly planning admit
a transparent economic interpretation. We identify three economic forces that shape the optimal stopping
rule. The first is a \emph{guarantee effect}--  the decision maker starts the search with a guaranteed amount (an endowment) which gets updated to the highest observed reward as the search progresses. A higher current guarantee shrinks the expected improvement from opening the next box and induces earlier
stopping. The second is a \emph{paid-scope effect}-- which is novel to our
setting. A larger prepaid scope lowers the marginal cost of accessing
future stages, raises the continuation value and supports continuation at
higher guarantees. The third is a \emph{remaining-horizon effect}-- as the
decision-maker advances through the sequence, the option value of continuing the search
shrinks because fewer future stages remain. The first and third forces
parallel the fall-back-value and deadline effects identified by
\citet{Leeli2022} in a different sequential-search environment. Most importantly, the
paid-scope effect has no counterpart in their setting or in the ordered-search
benchmark of \citet{boodaghians}, since both treat access to future boxes
as free. It is this second force, and its interaction with the other two,
that distinguishes our framework.

These three effects interact in economically meaningful ways. The paid scope effect attenuates the guarantee effect-- a broader prepaid scope raises the continuation value the most when the current guarantee is low; that is, precisely when the option value of future access matters the most. Submodularity of adjustment costs formalizes this complementarity and implies that a firm with broader current scope chooses (weakly) broader future scope. The remaining-horizon effect modulates the strength of the other two -- the paid-scope effect is the strongest early in the sequence, when many stages remain, while the guarantee effect dominates late, when scope differences matter less. As a result, the model generates both under-planning regimes, in which the firm begins with a narrow scope and expands adaptively; and over-planning regimes in which it commits broadly ex ante but stops early-- depending on primitives such as tail thickness and the curvature of planning costs. The examples in the paper illustrate these regimes and the distinct ways in which the forces interact.

In order to illustrate the essential idea of planning or implications of scope decisions, we start with a stylized example of search through an exogenously given sequence of options with costly planning. The example assumes a \textit{convex} scope cost and an incremental \textit{adjustment} cost for expanding scope.  Even in this simple setting, the example illustrates that the expected number of options searched is different from the initial optimal scope. This indicates that the initial planning might often change as the search process evolves. Thus, the example illustrates both the scope-dependence of reservation
values and the distinction between initial planning and realized search.

\subsection{An illustrative example}
\label{sec:example}
Consider a simple three-stage exploration problem in oil and
gas. Stage 1 is exploratory drilling, stage 2 is appraisal drilling,
conditional on the exploratory results; and stage 3 is the commitment to
development of infrastructure, conditional on the appraisal. Each stage
$i \in \{1,2,3\}$ yields a random reward $X_i$, interpreted as the present
value of recoverable reserves, equal to $H=100$ with probability
$p=1/2$ and $0$ otherwise, independently across stages. Inspecting each
stage costs $c=10$. The firm's outside option is $m_0=0$, and after
observing stages $1,\dots,i$, its guarantee is
$z_i=\max(m_0,X_1,\dots,X_i)$.

Before drilling begins, the firm chooses an initial scope
$j_0 \in \{1,2,3\}$ and incurs a scope cost $\phi(j_0)$, where
$\phi(0)=0$, $\phi(1)=10$, $\phi(2)=25$, and $\phi(3)=45$. The function
$\phi$ is increasing and convex. At stage $i$ with current scope $j$, the
firm can expand scope to $j'>j$ at cost
$\psi(i,j,j')=\phi(j'-i)-\phi(j-i)$, while  contraction of scope $j'\leq j$, is
costless.\footnote{The same structure can be used to describe a lab
developing a frontier AI model in stages---pretraining, post-training, and
deployment---with scope costs reflecting pre-committed compute, data, and
safety-evaluation capacity.} Suppose that the firm has pre-paid for the full scope $j=3$ and is at
stage $2$ with guarantee $z_2$. Continuing to stage 3 requires only the
inspection cost, so the continuation value is
\[
    -c+pH+(1-p)z_2 = 40+0.5z_2.
\]
Continuing is optimal if and only if $z_2<80$, yielding the reservation
value
\[
    r_2(3) = H-\frac{c}{p}=80.
\]

Next consider a firm that has pre-paid only for scope $j=2$. At stage 2,
accessing stage 3 requires expansion at cost $\psi(2,2,3)=10$, in
addition to the inspection cost. The continuation value becomes
\[
    -\psi-c+pH+(1-p)z_2 = 30+0.5z_2,
\]
so continuing is optimal if and only if $z_2<60$, yielding
\[
    r_2(2)=H-\frac{c+\psi}{p}=60.
\]

The expansion cost enters exactly as an additional inspection cost. As a
result, the two firms, identical in primitives and observed histories,
make different decisions. For example, if $z_2=70$, the firm with scope
$3$ continues, since $70<80$, while the firm with scope $2$ stops, since
$70>60$. The difference arises solely from the scope that has already been
paid for. The reservation value is therefore not dependent on the stage
alone, but also on scope as a state variable, and we write it as
$r_i(j)$.

Scope is {\it{not}} restricted to be an initial choice. A firm that begins with scope $2$ and
observes $X_1=0$ may expand to scope $3$ before stage 2, trading off the
cost of expansion against the benefit of a higher continuation threshold
at the next stage. The optimal policy therefore involves state-dependent
choices among continuing, expanding, contracting, and stopping. Scope also affects the initial planning decision. Working backward from
the stage-2 continuation values, the expected net payoffs from initial
scopes $1$, $2$, and $3$ are, respectively, $52.5$, $42.5$, and
$25$.\footnote{The calculation is as follows. At stage 2, the continuation
values at zero guarantee are $30$ under scope $2$ and $40$ under scope
$3$, since
\[
    -10-10+\frac{1}{2}\cdot 100 = 30,
    \qquad
    -10+\frac{1}{2}\cdot 100 = 40.
\]
Thus, after observing $X_1=0$, the continuation values at stage 1, where
$V(1,0,j)$ denotes the optimal continuation value after stage 1 with zero
guarantee and current paid scope $j$, are
\[
\begin{aligned}
    V(1,0,1)
    &= \max\left\{
        -10-10+\frac{1}{2}\cdot 100+\frac{1}{2}\cdot 30,\,
        -25-10+\frac{1}{2}\cdot 100+\frac{1}{2}\cdot 40
    \right\} = 45, \\
    V(1,0,2)
    &= \max\left\{
        -10+\frac{1}{2}\cdot 100+\frac{1}{2}\cdot 30,\,
        -15-10+\frac{1}{2}\cdot 100+\frac{1}{2}\cdot 40
    \right\} = 55, \\
    V(1,0,3)
    &= -10+\frac{1}{2}\cdot 100+\frac{1}{2}\cdot 40 = 60.
\end{aligned}
\]
Since the firm stops after stage 1 whenever $X_1=H$, the ex-ante payoff
from initial scope $j_0$ is
\[
    -\phi(j_0)-10+\frac{1}{2}\cdot 100
    +\frac{1}{2}V(1,0,j_0).
\]
Using $\phi(1)=10$, $\phi(2)=25$, and $\phi(3)=45$, this gives
$52.5$, $42.5$, and $25$ for $j_0=1,2,3$, respectively.}
Thus, the optimal initial scope is $j_0^*=1$. The firm commits at the
outset to the smallest possible scope, intending to expand if early
outcomes are unfavorable. Broader initial scopes are dominated in this
example: scope $3$ requires the firm to pre-pay $\phi(3)=45$ for stages
that it may never reach, while scope $2$ entails a higher upfront cost
without securing the full continuation value associated with scope $3$.

Under the optimal initial scope $j_0^*=1$, the firm inspects stage 1. If
$X_1=H$, it stops. If $X_1=0$, it expands to scope $2$ and inspects stage
2. If $X_2=H$, it stops; if $X_2=0$, it expands to scope $3$ and inspects
stage 3. Hence the stopping time $\tau^*$ takes values $1$, $2$, and $3$
with probabilities $1/2$, $1/4$, and $1/4$, respectively, such that the expected stopping time
$\mathbb{E}[\tau^*]=1.75$. The firm's expected stopping time exceeds its
initial scope, it pays initially for access to one stage but expects to
inspect more than one stage on average. 

This example illustrates an under-planning regime,
in which the firm rationally chooses a narrow initial scope and expands as
outcomes are realized. The difference between the initial optimal scope and expected search time (or number of options searched) is not surprising in this example, as the expansion of scope in the later stages comes at an {\it{additional}} cost.  But this example already illustrates the idea of optimal scope choice at any stage and adaptive decision-making. As we introduce an adjustment cost in the model, there would be a tension between choosing a longer scope and expected immediate outcomes. Section~\ref{sec:examples} shows that the opposite
regime, over-planning, where $\mathbb{E}[\tau^*]<j_0^*,$ also arises
under different primitives. We also derive comparative statics -- thresholds are
increasing in paid scope and decreasing in inspection and adjustment costs.
Thus, in contrast to the plan-neutral benchmark of \citet{weitzman} and
\citet{boodaghians}, the optimal stopping rule depends on past planning
decisions through the scope state.

\subsection{The economic effects and the example }
The calculations in the three-stage binary example make the interactions of the three economic forces transparent and exhibit why costly planning qualitatively changes search behavior. We recall that at stage $2$ with full prepaid scope ($j=3$), the continuation value is $40 + 0.5z_2$. The firm therefore continues if and only if $z_2 < 80$, which defines the scope-dependent reservation value $r_2(3) = 80$. When the firm has prepaid only up to scope $j=2$, accessing stage 3 requires an incremental adjustment cost of 10; the continuation value falls to $30 + 0.5z_2$ and the reservation value drops to $r_2(2) = 60$. Consequently, at an intermediate guarantee such as $z_2 = 70$ the firm with broader prepaid scope continues while the firm with narrower scope stops. This difference arises solely from the \emph{paid-scope effect} -- the firm that has already sunk the cost of access to stage $3$ faces a strictly lower marginal cost of continuation and is therefore willing to tolerate a higher current guarantee before stopping. The \emph{guarantee effect} is also visible -- had $z_2$ been $85$ instead of $70$, both firms would stop, because the current best reward already exceeds even the higher reservation value associated with the full scope. The remaining\emph{horizon effect} is implicit in the fact that only one stage remains after stage 2; the option value of continuation is therefore modest, which is why even the firm with $j=3$ stops once $z_2$ reaches 80.

These forces interact to generate the under-planning outcome observed in the example. Because $\phi$ is convex, committing ex ante to broad scope is expensive ($\phi(3)=45$ versus $\phi(1)=10$). The firm therefore optimally selects the narrowest initial scope $j_0^*=1$. After observing $X_1=0$ the realized guarantee is low, the guarantee effect is weak, and the paid-scope effect dominates -- the firm is willing to incur the expansion cost to raise its continuation threshold from $60$ to $80$. The same logic operates at stage 2 if $X_2=0$. The result is that the firm plans narrowly but expects to search $1.75$ stages on average — precisely the pattern of under-planning that costly scope planning naturally produces. In the absence of the paid-scope channel, i.e., if access to future stages were free or independent of planning history, as in \cite{weitzman} and \cite{boodaghians}, both firms would face the identical reservation value at stage $2$ and there would be no economic reason for an initial narrow plan followed by systematic expansion after poor early draws. The interaction of the three forces with the convexity of $\phi$ is what produces the realistic adaptive behavior documented in the example and, more generally, in R\&D and resource-exploration settings. Section $3$ shows that the opposite (over-planning) regime also arises under different primitives, again through the interplay of the same three forces.

\subsection{Summary of results}

We establish three main results in Section~\ref{sec:results}. First, the
optimal policy is a threshold rule indexed by the firm's paid scope. At
any state $(i,z,j)$, consisting of stage $i$, guarantee $z$, and paid
scope $j$, the firm stops if and only if $z \geq r_i(j)$ and continues
otherwise. Conditional on continuing, the firm chooses an optimal next
scope $j' \in \{i+1,\dots,n\}$, trading off the adjustment cost of
expansion against the option value of access to later stages. The
thresholds $r_i(j)$, which we call scope-dependent reservation values,
exist, are unique, and are finite under standard integrability and cost
assumptions.

Second, we derive comparative statics. The threshold $r_i(j)$ is weakly
increasing in paid scope: a firm with broader prepaid scope is willing to
continue at higher guarantees because future stages are cheaper to access.
The threshold is weakly decreasing in inspection and adjustment costs.
Under a submodularity condition on adjustment costs, the optimal next
scope is monotone in current scope, so a firm with broader current scope
chooses weakly broader future scope.

Third, the model nests the standard ordered-search benchmark when planning
costs vanish. Section~\ref{sec:block-adaptive} studies this benchmark
through an alternative formulation, which we call \emph{scope-adaptive
search}. In this formulation, payoffs are defined recursively over scopes of different lengths, with adaptive revision built into the payoff structure rather
than only into the policy. When planning is costless, the scope-adaptive
payoff coincides with the dynamic-programming value function, and the
reservation values under the two formulations are identical. With nonzero
scope and adjustment costs, paid scope becomes a state variable, and the
scope-dependent thresholds generally differ from the plan-neutral dynamic
programming benchmark.

Section~\ref{sec:examples} works out two additional cases. The first example (Section~\ref{sec:normal}) studies normally distributed
rewards with convex scope and adjustment costs. It provides sufficient
conditions under which the firm never expands on the optimal path. Under
these conditions, the expected stopping time is bounded above by the
initial scope, $\mathbb{E}[\tau^*]\leq j_0^*$, which we interpret as an
over-planning regime. The second example (Section~\ref{sec:tails}) compares thin-tailed
(exponential) and heavy-tailed (Pareto) reward distributions matched on
mean. The cost of incomplete scope, measured by the gap $r_2(3)-r_2(2)$,
is constant under exponential rewards but becomes large under sufficiently
heavy-tailed Pareto rewards. Heavy-tailed reward distributions therefore
generate an endogenous force toward broader initial planning where the option
to access future stages is more valuable when every stage carries a
meaningful chance of a large breakthrough.

\subsection{Related Literature}

The literature on search is expansive. While there are several papers focusing on the classical model by \cite{weitzman}, the literature lately has moved towards applied models with realistic features including pre-ordered options (\cite{choi}). While \cite{boodaghians} discusses the optimal search strategy in a more generalized ordered setting, the paper does not specifically delve into economic applications. We aim to model sequential search with pre-ordered options where planning plays a very crucial role. We explore the implications of costly planning and the implications of choosing a horizon for such planning. One of the most important aspects of our model is that expansion of horizon always remains feasible at an additional (adjustment) cost. We analyze subsequent adaptive behavior that results from optimal search strategy in such a setting. Our approach aims to fill an important gap in the existing literature and it is likely to provide insights and modeling tools not only in individual search behavior in managerial operations and decision-making but also in understanding equilibrium behavior of competing firms where planning for horizon or stages remains a key strategic variable.

We find several areas of applications for our setting. We broadly categorize the related papers into the following three areas:

\subsubsection{Search models, Weitzman's rule and other variants}
\begin{enumerate}
    \item[(i)] \citet{Roberts1981-yd} studies two variants of sequential
    development projects. Their two-sided sequential development project,
    in which the DM can adopt the project at any intermediate stage based
    on the current value estimate, shares the stop-or-continue structure
    of our model. Their optimal policy is a stage-dependent threshold rule,
    in which the DM continues when the current estimate falls in an
    intermediate range. The continuation threshold becomes $r_i(j)$, and it depends
    not only on the stage but on the DM's planning history through the
    scope state $j$. Roberts and
    Weitzman characterize when to stop given that all stages are freely
    accessible. We characterize when to stop, when to expand, and how much
    to commit to in advance, given that access itself is costly.

    \item[(ii)] \citet{doval} extends \citet{weitzman} by allowing the
    agent to take any uninspected box without paying its inspection cost.
    This introduces a backup value $x_i^B$ alongside the reservation value
    $x_i^R$, where below $x_i^B$ the box is taken without inspection,
    above $x_i^R$ it is ignored, and in between it is inspected. The same
    box may receive different treatment depending on the history of
    realized prizes. We
    replace the pair $(x_i^B,x_i^R)$ with a family of thresholds $r_i(j)$
    indexed by the scope state $j$, where the DM's planning history
    determines which threshold applies. \cite{doval} enriches what the DM can do at the moment of
    stopping, whereas we enrich what the DM must pay for in advance to be eligible to search at all. 

\item[(iii)] \cite{Leeli2022} discusses a very important aspect of sequential search, that is, choice of intensity of search and adaptive behavior when there is a deadline. The paper characterizes optimal search strategy and provides comparative statics results with respect to change in default values (fall-back value effect) and its interactions with deadline effect. The fall-back-value and deadline effects identified by \citet{Leeli2022} in a different sequential-search environment parallel two of the forces in our model. Most importantly, the paid-scope effect has no counterpart in their setting, since both treat access to future boxes as free. Their deadline is a hard wall in calendar time; ours is mechanical in the stage index, with each inspection consuming one of a fixed number of stages.
\end{enumerate}

\subsubsection{Pre-ordered search}
There is a growing literature on sequential search in richer environments that feature pre-ordered alternatives, endogenous horizons, or adaptive scope. In \cite{boodaghians}, the generalized reservation value
    $z_i$ for each stage on a line is determined by the costs and
    distributions of all downstream boxes. It is a fixed number,
    independent of the search history. In our model, the reservation value
    $r_i(j)$ at stage $i$ depends on the current scope $j$, which is itself
    the outcome of prior planning decisions. The interval $[r_i(i),r_i(n)]$
    measures the extent to which the planning margin affects the search
    decision at stage $i$. When this interval is degenerate, scope is
    irrelevant and the model reduces to \citet{boodaghians}. We add the
    dimension of paid eligibility and show that the firm's planning history
    becomes a state variable that indexes a family of reservation values
    rather than a single one.

 \cite{choi} develops a tractable extension of \cite{weitzman}'s Pandora’s-box framework to sequential search among pre-ordered options, with applications to web and informed search. They analyze non-stationary payoffs, recall, and option value while preserving indexability. \cite{mojir} nests a finite-horizon store-by-store search problem inside an infinite-horizon temporal search problem, highlighting interactions across different layers of horizon choice. \cite{pavan} studies sequential search in which agents endogenously select their consideration sets, thereby treating the scope of search as an active choice variable. \cite{Leeli2022} examine adaptive search intensity under deadlines and derive fall-back-value and deadline effects that parallel two of the forces that we also identify in our model. These papers enrich the ordered-search benchmark of \cite{boodaghians} by allowing more flexible adaptation over horizons and sets of alternatives. However, they do not introduce costly ex-ante planning that renders paid scope a persistent state variable. As a result, continuation values and reservation thresholds remain independent of planning history, and the distinctive paid-scope effect that arises in our setting is absent. 

\subsubsection{Industrial organization}

In the industrial-organization literature, firms competing in sequential technological development routinely choose both the timing and the scope of their R\&D investments. Classic multi-stage patent-race models (\cite{REINGANUM1989849}; \cite{harrisvickers}) analyze strategic effort and leadership dynamics along a technological path whose length is effectively chosen by participants. Step-by-step innovation frameworks (\cite{aghion2001}) emphasize how technological distance between rivals shapes innovation incentives, while the design of patent scope and breadth in cumulative innovation settings directly affects the returns to sequential advances (\cite{Klemperer}). These contributions highlight the strategic importance of horizon and scope choices under competition, yet they treat access to future stages as free or independent of prior planning commitments. By introducing costly ex-ante planning, our model makes paid scope a state variable that shapes continuation values even for decision makers who share the same stage and current guarantee. This generates scope-dependent reservation values $r_i(j)$ and the novel paid-scope effect that moderates the guarantee effect while interacting with the remaining-horizon effect to produce both under-planning and over-planning regimes—features that are new to both the ordered-search and multi-stage R\&D literatures.

The paper is organized as follows. Section~\ref{sec:model} introduces the
formal model, including the firm's outside option and the three cost
components: inspection, scope, and adjustment costs.
Section~\ref{sec:results} states the main results, which establish
existence and uniqueness of scope-dependent reservation values, the
threshold characterization of the optimal policy, and the comparative
statics. Section~\ref{sec:block-adaptive} develops the scope-adaptive
formulation of the costless-planning benchmark and establishes its
equivalence with dynamic programming. Section~\ref{sec:examples} works out
two illustrative examples with
Section~\ref{sec:normal}
covering the normal-rewards specification with convex adjustment costs and
Section~\ref{sec:tails} comparing thin- and heavy-tailed reward
distributions. Section~\ref{sec:conclusion} concludes. Proofs are collected in the Appendix.
\section{The Model}
\label{sec:model}

\noindent Consider a finite set of $n$ boxes that a DM can plan and explore in a fixed sequence. We index the boxes as $ \mathcal N = \{1, 2, \dots, n\},$ where $i\in \mathcal{N}$ denotes the position of the box in the sequence. Each box $i$ contains a random reward $X_i\sim F_i$ and opening it costs $c_i \geq 0$. The DM plans to search the boxes in the given sequence and cannot skip a box. That is, in order to observe $X_i$, all preceding boxes $1,2,\ldots ,i-1$ must be observed. The DM is risk-neutral and maximizes expected payoff. We assume that during the search, irrespective of the initial plan, the search can be terminated by the DM at any $i\leq n$.  Let $m \in \mathbb{R}$ denote an amount that is guaranteed. Before the search begins, the DM has an initial endowment $m_0,$ i.e., $m$ is $m_0$ initially. The amount $m$ is then updated with each observed reward: after box $i$ is observed, $m_{i}=\max\{m_{i-1},x_{i}\}.$\medskip

\noindent We assume that for each $i\in \mathcal{N}$, the distribution \(F_i\) is known to the DM, and the random variables \(X_i\) are independent across \(i \in \mathcal N \).\footnote{Allowing for correlation across $X_{i}$ would complicate the dynamic problem by introducing informational spillovers across boxes. Our focus on independence isolates the role of reservation values; extensions to correlated environments are left for future work.} For each \(i \in \mathcal N \), the distribution \(F_i\) is defined on \(\mathbb{R}\), thereby allowing for both positive and negative payoffs, although in many applications rewards are typically non-negative. The cumulative distribution function \(F_i(x) = \mathbb{P}(X_i \leq x)\) is non-decreasing and right-continuous, and satisfies \(F_i(-\infty) = 0\) and \(F_i(\infty) = 1\).\medskip

\noindent We introduce the following assumptions.

\begin{assumption}[Integrability]\label{ass:integrability}
For each $i\in \mathcal N $, $E\left[|X_i|\right] < \infty$. Consequently, all finite maxima of $X_i$ have finite expectation.
\end{assumption}
\noindent Integrability ensures that expectations involving $X_i$, such as $E\left[\max(m_{i-1}, X_i)\right]$, 
are well-defined and finite. This assumption is also important for our dynamic programming approach: for instance, for any box $i$, denote by $V_{i+1}(\max(m,X_i))$ the value function for the dynamic program that solves the optimal search problem. Assumption \ref{ass:integrability} then ensures that $E\left[V_{i+1}(\max(m_{i-1}, X_i))\right]$ is well-defined and enables the application of standard results such as the dominated convergence theorem (DCT) and Fubini's theorem.
\begin{assumption}[Inspection costs]\label{ass:strictcost}
For each $i\in \mathcal N $, the cost of inspection satisfies $0<c_i<\infty$.
\end{assumption}

\noindent Assumption \ref{ass:strictcost} requires that inspection is strictly costly for every box. If inspection were costless for a box, the DM would always open it. Positive inspection costs make the stopping decision non-trivial and are central to many economic applications. In particular, when the current guarantee $m$ is sufficiently large, the DM may optimally choose not to continue searching. 

The formulation and revision of plans often entails costs. The DM may invest in creating a plan at the beginning of the search process and revision during the search may also require additional cost. For instance, \cite{Roberts1981-yd} discusses investment in a seismograph test prior to undertaking oil exploration in an area. Such expenses provide information about expected rewards from the search and may influence the search decisions. Further, during the search process, a change of plans i.e., adjusting or revising the planned scope may also be costly. In order to allow for these possibilities, we generalize the above model by introducing a \textit{scope cost} $\phi$, which captures the cost of ex-ante planning and an \textit{adjustment cost} $\psi$, which is incurred if a planned scope is revised during the search. 

The interaction between inspection costs and scope costs is central to the planning problem we study. The scope cost enables the DM to determine how far ahead in the sequence they can search. Specifically, before the search begins, the DM chooses an initial scope $j_0 \in \{1,\ldots,n\}$ and pays a cost $\phi(j_0)$. This grants access to the boxes $\{X_1,\ldots ,X_{j_0}\}$, $j_0\in \{1,2,\ldots ,n\}$. At stage $i$ in the search, the DM may revise the scope to any $j\geq i$.
 
\begin{assumption} [Scope costs]
We assume that the scope cost function $\phi:\{0,1,\ldots,n\}\to\mathbb{R}_+$ satisfies the following properties: {\normalfont(a)} $\phi(0)=0$ and {\normalfont(b)} $\phi$ is strictly increasing i.e. $\phi(s+1)>\phi(s)$ for every $s=0,\ldots,n-1$.
\end{assumption}

\noindent Suppose that the DM has already paid $\phi(j)$ when they arrive at stage $i$. We refer to $j$ as the \emph{current scope}: the DM has paid to keep open the possibility of searching through stage $j$. Boxes beyond $j$ are not directly accessible but can be made accessible through the expansion of the scope by paying an adjustment cost, $\psi.$ In particular, if the current accessible set is $\{X_1,\ldots,X_j\}$, the DM can expand the scope to $\{1,\ldots,j'\}$ with $j' > j$ by incurring an additional \textit{adjustment} cost $\psi(i, j,j')$.  

\begin{assumption}[Adjustment costs]
\label{ass:adjustment}
We assume that at stage $i$ with current scope $j$, if the DM chooses a new scope
$j' \in \{i+1,\ldots,n\}$, the adjustment cost is given by $\psi(i, j, j^{'})$.  Adjustment cost is $0$ when $j^{'}\le j$  and is positive if $j^{'}> j.$
\end{assumption}

The search proceeds as follows -- the DM selects an initial planned scope $j_0$ and initiates the search if $j_0\geq 1$. In order to
open the $i$th box the DM must (i) observe all the preceding boxes $\{X_1, X_2, \ldots, X_{i-1}\}$; and (ii) pay the inspection cost $c_i$. After observing $X_i$, she decides whether to terminate the search, or to continue. If the DM decides to continue, they may choose a revised scope. We index scope choices by the stage at which they are made i.e. $j_0$ denotes the initial scope, chosen before any box is
opened, and $j_i$ for $i \geq 1$ denotes the scope chosen at the stage-$i$ (after observing $X_i$, conditional on continuing).  A sequential search strategy $\pi$ specifies  $j_0$, an action (stop/continue) at each stage $i$ and if the DM continues, a revised scope $j_i$. Since the DM can revise the scope at any stage, the strategy is \textit{adaptive}. Let $\Pi$ be the set of all such adaptive strategies. Formally, a sequential search strategy $\pi\in \Pi$ specifies:
\begin{enumerate}
    \item[(i)] an initial scope $j_0 \in \{1, \ldots, n\}$, with cost
    $\phi(j_0)$ paid before any box is opened;
    \item[(ii)] at each stage $i \in \{1, \ldots, n-1\}$, given the current
    scope $j_{i-1}$ and the current best reward
    $y_i = \max(m_0, x_1, \ldots, x_i)$, whether to
    \begin{enumerate}
        \vspace{-1em}
        \item \normalfont{Stop:} accepting $y_i$ and terminate the search; or 
        \item \normalfont{Continue:} choose a new scope
        $j_i \in \{i+1, \ldots, n\}$, pay the adjustment cost
        $\psi(i, j_{i-1}, j_i)$ and the inspection cost $c_{i+1}$ to observe $X_{i+1}$.
    \end{enumerate}
\end{enumerate}

Each
$\pi \in \Pi$ induces a stochastic stopping time $\tau^{\pi}$ and a
stochastic sequence of scope choices $(j_0, j_1, \ldots, j_{\tau^{\pi}-1})$,
both adapted to the sequence of realized values. Let $\mathcal{X}(\pi) = \{1, 2, \ldots, \tau^{\pi}\}$ denote
the set observed under $\pi$, with $\mathcal{X}(\pi) = \emptyset$ if the DM
stops at the outset. 
\noindent The DM chooses $\pi$ to maximize
\begin{equation}
\pi^* \in \arg\max_{\pi \in \Pi}
\mathbb{E}\!\left[
    \max\!\left\{m_0, \max_{i \in \mathcal{X}(\pi)} X_i\right\}
    \;-\; \sum_{i \in \mathcal{X}(\pi)} c_i
    \;-\; \phi(j_0)
    \;-\; \sum_{i=1}^{\tau^{\pi}-1} \psi(i, j_{i-1}, j_i)
\right]
\end{equation}
with the convention that if $\mathcal{X}(\pi) = \emptyset$, the reward equals $m_0$
and no inspection or adjustment costs are incurred. The total planning expenditure under $\pi$ is
$\phi(j_0) + \sum_{i=1}^{\tau^{\pi}-1} \psi(i, j_{i-1}, j_i)$, with the
sum empty if $\tau^{\pi} \leq 1$.  

\noindent We define the DM's continuation value  recursively with the scope $j$ as an explicit state variable. The problem admits a DP formulation with state $(i,z,j)$, where $i$ is the current stage, $z$ is the current guarantee, and $j$ is the current paid scope. The value is evaluated after observing stage $i$ and updating the guarantee to $z=m_i$.

\begin{definition}[Sequential search dynamic program (SSDP)]
\label{def:bellman}
Let $V(i,z,j)$ denote the value after observing stage $i$ and updating the guarantee to $z$.  At the initial stage $i=0$,
\begin{equation}
\label{eq:dp}
    V_0(m_0) = \max\!\left\{m_0,\;
    \sup_{j_0 \in \{1,\ldots,n\}} \left[-\phi(j_0) - c_1
    + \mathbb{E}\!\left[V(1, \max(m_0, X_1), j_0)
    \right]\right]\right\}.
\end{equation}
For $i<n$,
\begin{equation}
\label{eq:dp}
V(i,z,j)
=
\max\!\left\{
z,\;
\sup_{j' \in \{i+1,\ldots,n\}}
\left[
-\psi(i,j,j') - c_{i+1}
+ \mathbb{E}\!\left[V(i+1,\max(z,X_{i+1}),j')\right]
\right]
\right\}
\end{equation}
For $i=n$, the value function is $V(n,z,j)=z$ \text{for all feasible }j.
\end{definition}

\noindent We refer to an SSDP where Assumptions \ref{ass:integrability}-\ref{ass:adjustment} hold as a \textit{sequential search problem} (SSP) for brevity. 

\subsection{Optimal sequential search strategy}
\label{sec:results}
\noindent In the standard Pandora’s Box problem in \cite{weitzman}, there is no exogenously given order: the DM chooses the sequence in which boxes are searched. The optimal search policy is characterized by a reservation value $\zeta_i$ for each box $i$, and the optimal sequence is given by the decreasing order of $\zeta_i$. The reservation value $\zeta_i$ is defined as the amount that makes the DM indifferent between stopping and opening box $i$, and depends only on the distribution $F_i$ and the inspection cost $c_i$. A key feature of this formulation is that $\zeta_i$ is independent of the previously observed rewards as well as the set of remaining boxes.  In our model, the key solution concept of deciding whether to stop or continue the search at any stage based on the reservation values is preserved. However, since the boxes can be opened only in a given sequence and the DM incurs costs to make future boxes accessible in addition to the inspection cost, the reservation values are no longer independent of the sequence and the selected scope. The reservation values will depend on the current scope $j$ and their calculation would consider both (i) the sequential structure of the problem and (ii) planning and adjustment costs in addition to $c_i$.

 \noindent In contrast to both \cite{weitzman} and \cite{boodaghians}, reservation values in our model are no longer state-independent. Since access to future boxes is endogenous and costly, continuation decisions depend on previously chosen scope. This leads to a family of \emph{scope-dependent reservation values}.

\begin{definition}[Scope-dependent reservation value]
\label{def:reservation}
For each stage $i<n$ and feasible current scope $j\in\{i,\ldots,n\}$, the scope-dependent reservation value $r_i(j)$ is defined as the unique solution to
\begin{equation}
\label{eq:reservation}
    V(i, r_i(j), j) = r_i(j).
\end{equation}
Equivalently, $r_i(j)$ is the amount at which the DM is indifferent at stage $i$ between stopping and continuing optimally, given current scope $j$.
\end{definition}

 \begin{definition}[Optimal strategy]
A strategy $\pi$ is optimal if, at every feasible state $(i,z,j)$, it attains the value function $V(i,z,j)$. That is, at each state, $\pi$ prescribes either stopping, which yields payoff $z$, or continuing with a scope choice $j' \in \{i+1,\ldots,n\}$ such that $j'$ attains
\[
\max_{j' \in \{i+1,\ldots,n\}}
\left[
-\psi(i,j,j') - c_{i+1}
+ \mathbb{E}\!\left[V(i+1,\max(z,X_{i+1}),j')\right]
\right].
\]
Note that $i=0$ and $j'=j_0$ in the above when the search begins.
\end{definition}

\noindent At any stage $i < n$, given current guarantee $z$ and scope $j$, the DM chooses between stopping and continuing. If the DM stops, the payoff is the current guarantee $z$. If the DM continues, a new scope $j' \in \{i+1,\ldots,n\}$ is chosen with the adjustment cost $\psi(i,j,j')$ and the inspection cost $c_{i+1}$ is paid to observe $X_{i+1}$. The guarantee updates to $\max(z, X_{i+1})$, and the continuation value is given by $V(i+1, \max(z, X_{i+1}), j')$. The DM selects the scope $j'$ to maximize expected continuation value net of costs. This recursive structure yields the Bellman equation in Definition \ref{def:bellman}. \medskip

\begin{theorem}\label{thm: optimalstrategy}
An optimal strategy $\pi^*\in \Pi$ exists for an SSP. 
\end{theorem}

\noindent The result follows from a standard backward induction argument for finite-horizon dynamic programs. The assumptions ensure that continuation values are well-defined and that the maximization over scope choices is attained at each stage.  

 \begin{definition}[Planned sequential search strategy (PS)]\label{defn: optimal strategy}
\noindent Before the search begins, the initial scope $j_0$ is chosen according to
\[
j_0(z) \in \arg\max_{j \in \{1,\ldots,n\}}
\left[-c_1 - \phi(j) + \mathbb{E}\!\left[V(1, \max(z, X_1), j)\right]\right].
\]

\noindent At any stage $i \in \{0,1,\ldots,n\}$ with current scope $j$ and guarantee $z$:
\begin{enumerate}
    \item[(i)] The DM stops if and only if $z \ge r_i(j)$.
    \item[(ii)] If $z < r_i(j)$, the DM continues, and the optimal next scope is given by
    \[
        j'_i(z, j) \in \arg\max_{j' \in \{i+1,\ldots,n\}}
        \left[
        -\psi(i,j,j') - c_{i+1}
        + \mathbb{E}\!\left[V(i+1, \max(z, X_{i+1}), j')\right]
        \right].
    \]
\end{enumerate}

\end{definition}

\noindent Unlike the standard Pandora’s Box problem, where reservation values depend only on the box, our thresholds $r_i(j)$ depend on the state variable $j$ (the current paid scope). This introduces a complication: the DM’s stopping rule changes depending on how much scope has been prepaid. A DM with a broader prepaid scope continues longer (i.e., has a higher $r_i(j)$), because future stages are cheaper to access. Thus, the same box may have different thresholds depending on the paid scope in addition to its position in the sequence. This contrasts with \cite{boodaghians}, where reservation values depend on the downstream structure but remain independent of past decisions.

\begin{theorem}[Existence and uniqueness]
\label{thm:reservation-exists}
In an SSP, for every stage $i$ and every feasible scope $j$, the scope-dependent reservation value $r_i(j)$ exists, is finite, and is unique.
\end{theorem}

\noindent We now state our main result.

\begin{theorem}\label{thm: optimal strategy}

\noindent PS is the optimal strategy for any SSP.  
\end{theorem}

\noindent So far, we assumed that once committed to a scope $j,$ the DM faces no costs when terminating the plan before $j$ while  a positive adjustment cost needs to be paid to expand the scope of search. The DM will have to pay a differential cost that accounts for the newer scope.\footnote{We do not impose further structure on $\psi$ unless stated otherwise. When $\psi \equiv 0$, scope expansion is costless, so the DM can make all boxes accessible without incurring any planning cost.} Thus, the DM’s total cost consists of inspection costs, scope cost and the cost of adjusting the scope. Next, we model adjustment costs as submodular in scope $j$ and show some interesting implications on the stop/search decisions of the same. 
\begin{assumption}[Submodular adjustment costs]
\label{ass:submodular}
For each stage $i$, the adjustment cost function
$\psi(i,\cdot,\cdot)$ is submodular in $(j,j')$: for all
$j_1 \le j_2$ and $j_1' \le j_2'$ in the feasible set,
\[
    \psi(i,j_1,j_2') - \psi(i,j_1,j_1')
    \;\ge\;
    \psi(i,j_2,j_2') - \psi(i,j_2,j_1').
\]
Equivalently, the net-of-adjustment payoff $-\psi(i,j,j')$ has
increasing differences in $(j,j')$.
\end{assumption}
\noindent This assumption captures the idea that having a larger current scope reduces the marginal cost of further expansion.\medskip

\begin{proposition}[Comparative Statics]
\label{prop: CS}
The following are true.
\begin{enumerate}
    \item[\normalfont(i)] $V(i,z,j)$ is weakly increasing in the guarantee $z$ and weakly increasing in the scope $j$. The reservation value $r_{i}(j)$  is weakly increasing in $j;$ that is, if $j' \ge j$, then $r_i(j') \ge r_i(j)$.

    \item[\normalfont(ii)] Let $V^c$ denote the value function under inspection cost $c_{i+1}$, and let $r_i^c(j)$ denote the associated reservation value. If $\tilde c_{i+1} \ge c_{i+1}$, then
    $r_i^{\tilde c}(j) \le r_i^c(j)  \text{ for every feasible } j.$
    Let $V^\psi$ denote the value function under adjustment cost function $\psi$, and let $r_i^\psi(j)$ denote the associated reservation value. If
    $\tilde \psi(i,j,j') \ge \psi(i,j,j')  \text{ for all feasible } (i,j,j'),$
    then
    $r_i^{\tilde \psi}(j) \le r_i^\psi(j) \text{ for every feasible } j.$

    \item[\normalfont(iii)] Under Assumptions~\ref{ass:integrability}--\ref{ass:submodular}, fix a stage $i$ and guarantee $z$. Define the optimal next-scope correspondence,
    $$J_i(z,j)
    =
    \arg\max_{j' \in \{i+1,\ldots,n\}}
    \left\{
    -\psi(i,j,j') - c_{i+1}
    + \mathbb E\!\left[V(i{+}1,\max(z,X_{i+1}),j')\right]
    \right\}.
    $$
    Then $J_i(z,j)$ is weakly increasing in $j$ in the strong set order.\footnote{A correspondence is weakly increasing in the strong set order if its minimal and maximal selections are weakly increasing.}
    In particular, $\underline{j}_i(z,j) = \min J_i(z,j)$ and $\overline{j}_i(z,j) = \max J_i(z,j)$ are weakly increasing in $j$.
\end{enumerate}
\end{proposition}

\noindent Proposition \ref{prop: CS} highlights the economic role of scope in shaping search behavior. Part (i) shows that a higher guaranteed payoff or a larger prepaid scope increases the value of continuing the search. In particular, when more scope has already been paid for, the DM is willing to continue even when the guaranteed amount is high, which leads to higher reservation values.  Part (ii) shows that higher inspection or adjustment costs make continuation less attractive, lowering reservation values. Finally, part (iii) shows that under submodular adjustment costs, the optimal scope choice is monotone: a DM with a larger current scope chooses (weakly) larger future scope. This reflects the fact that expanding scope becomes less costly at the margin when more scope has already been prepaid.

\noindent Next, we find an upper bound on the probability that a DM following the general optimal search strategy expands the scope during the search i.e. stops the search at a time later than the pre-committed scope $j_0^*$.

\begin{proposition}[Probability of adaptive expansion]
\label{prop:expansion-bound}
Let $j_0^*$ denote the optimal initial scope and $\tau^*$ the optimal
stopping time. Under Assumptions~1--4,
\begin{equation}
\label{eq:bound}
\Pr(\tau^* > j_0^*) \;\leq\; \mathbf{1}\{m_0 < r_{j_0^*}(n)\} \cdot
\prod_{i=1}^{j_0^*} F_i\!\left(r_{j_0^*}(n)^{-}\right),
\end{equation}
where $F_i(x^-) = \Pr(X_i < x)$.
\end{proposition}

Proposition~\ref{prop:expansion-bound} bounds the probability that the firm
expands its planning horizon mid-search, expressed in terms of model
primitives and the optimal initial scope. The bound is the probability that
the firm fails to encounter, across its $j_0^*$ pre-paid stages, a reward
high enough to make stopping optimal at the broadest feasible scope. It is
sharp in the sense that any tighter bound must use information about the
firm's actual scope path, which depends on the realized rewards. Three
features make the result useful. First, the right-hand side is computable in
primitives -- the reservation value $r_{j_0^*}(n)$ is determined by the
downstream cost structure and reward distributions, and the product
involves only the marginal distributions $F_1, \dots, F_{j_0^*}$. Second,
the bound delivers immediate qualitative implications. It is zero when
$j_0^* = n$, recovering the intuition that a firm pre-committed to its
full horizon has nothing to expand into; it is decreasing in $j_0^*$, so
broader pre-commitment leaves less room for adaptive deviation; and it is
monotone in inspection and adjustment costs through their effect on
$r_{j_0^*}(n)$.\footnote{Furthermore, the bound provides a probabilistic counterpart to
the over-planning result of Section~3.2, which establishes
$\mathbb{E}[\tau^*] \leq j_0^*$ under sufficient conditions: when those
conditions tighten so that no expansion is ever optimal on the equilibrium
path, the bound collapses to zero. The example of Section~1.1, by contrast,
exhibits the under-planning regime in which the bound is positive and
tight: there, the bound equals $\Pr(X_1 = 0) = 1/2$, which coincides with
the actual probability of expansion.} 

\subsection{Sequential search with costless planning}
\label{sec:block-adaptive}
In this section, we consider the case where planning does not require any ex-ante cost i.e $\phi(j_0)=0$ and $\psi(i,j,j')=0$ for all $i,j,j'$. The DM only pays the inspection cost for observing a box and can revise the plan without additional cost during the search. Setting $\phi(j_0)=0$ and $\psi(i,j,j')=0$ in SSP, we get \cite{weitzman}'s Pandora's box problem with boxes arranged in an exogenous sequence. This is similar to the structure in \cite{boodaghians}. We note that planning has no implications for the search/stop decisions in the absence of scope and adjustment costs since the DM can costlessly revise a plan at any stage. We introduce an alternate formulation that explicitly models the search procedure that may be obscured by the general framework of the DP formulation, and we show that the optimal scope in the absence of planning costs is $n$. 
\noindent Consider any $i < n$ and current scope $j \in \{i,\ldots,n\}$. After observing
$X_i = x$, let $z = \max(m, x)$ denote the maximum guaranteed reward (maximum of endowment $m$ or the best observed reward $x$) at stage $i$. We introduce the following recursive function as an alternative way of modeling the sequential search problem.
For any scope $i,\ldots ,j$ and any realised value $x$ of $X_i$, the agent's maximum payoff is 
$H^\phi_{j,j}(m, x) = \max(m, x),$
in the terminal case $i=j$. For all the preceding cases, where  $i<j,$
\begin{equation}
\label{eq:adaptive}
    H^\phi_{i,j}(m, x) = \max\!\left\{z,\;
    \sup_{j' \in \{i+1,\ldots,n\}} \left[-\psi(i,j,j') - c_{i+1}
    + \mathbb{E}\!\left[H^\phi_{i+1,j'}(z, X_{i+1})
    \right]\right]\right\}.
\end{equation}

\noindent We call the above payoff function the \textit{scope adaptive payoff}. $H^\phi_{i,j}(m,x)$ models  the sequential search problem without using the standard DP value function. Its recursive construction makes it adaptive to the plan being revised during the search. The optimal revision of scope choice from $j$ to $j'$ is  modeled by the following expression, \[\sup_{j' \in \{i+1,\ldots,n\}} \left[-\psi(i,j,j') - c_{i+1}
    + \mathbb{E}\!\left[H^\phi_{i+1,j'}(z, X_{i+1})
    \right]\right]\]. 
   \noindent Given the above formulation, what is the optimal scope that a DM would choose? When $\psi(i,j,j')=0$ for all $i,j,j'$, it is easy to show that the payoff $H^\phi_{i,j}(m,x)$ is weakly increasing in $j,j'$. This monotonicity property of the scope adaptive payoffs implies that it is always optimal for a DM to choose the scope $j_0=n$ and subsequently $j_i=n$ at each stage $i$ when the search continues. If it is costless to plan to explore any stage $k'\geq j$, the decision to open the $k$th box would depend only on its expected reward, inspection cost and the expected rewards from the subsequent $k+1,\ldots n$ unopened boxes. The choice of scope in this case becomes redundant. Further, we find that the optimal strategy in both the formulations, DP and scope adaptive formulation (SAF) are equivalent. We state this result below and provide the proof in the appendix.\footnote{In contrast, when $\psi \neq 0$ the ex-ante value of the initial
scope need not be monotone in $j$: the example of
Section~\ref{sec:example} already exhibits a strictly interior optimum
$j_0^* = 1$, with broader scopes dominated because the convex scope cost is
paid up front for stages that may never be reached.} 
    \begin{theorem}\label{thm: equivalence DP and SAF}
        The optimal search strategy in SSP and in SAF are equivalent for all $i\leq n$. The optimal scope  when $\psi=0$ is $n$. \end{theorem}




\subsection{Economic effects governing the optimal strategy}

The threshold characterization in Theorem~3 and the comparative statics in
Proposition~1 are driven by three economic forces, which we now make
explicit. Define the net continuation payoff,
\begin{equation*}
  f_i^{\,j}(z) \;:=\; C_i^{\,j}(z) - z,
\end{equation*}
where
\begin{equation*}
  C_i^{\,j}(z)
  \;=\;
  \sup_{j' \in \{i+1,\ldots,n\}}
  \Bigl\{
    -\psi(i, j, j') - c_{i+1}
    + \mathbb{E}\bigl[ V(i+1, \max\{z, X_{i+1}\}, j') \bigr]
  \Bigr\},
\end{equation*}
is the optimal continuation payoff at state $(i, z, j)$. The reservation
value $r_i(j)$ is the unique solution to $f_i^{\,j}(r_i(j)) = 0$
(Theorem \ref{thm:reservation-exists}). Next we formalize the economic effects which are intertwined and affect the optimal search strategy.

\subsubsection{Guarantee effect}

Consider any stage $i < n$ and any feasible paid scope $j$. The net continuation payoff $f_i^j(z)$ is continuous and \emph{strictly decreasing} in the current guarantee $z$ on the set $\{z : f_i^j(z) \ge 0\}$. Further,
$\lim_{z\to -\infty} f_i^j(z) = +\infty$ and $ \lim_{z\to +\infty} f_i^j(z) < 0.$
Consequently there exists a unique root $r_i(j)$, and the decision maker stops if and only if the realized guarantee reaches or exceeds this threshold. We note  that the value functions satisfy the Lipschitz bound $0 \le V(i+1,w_1,j') - V(i+1,w_2,j') \le w_1 - w_2$ for $w_1 \ge w_2$ (finite-horizon backward induction). Hence for any $h > 0$, $C_i^j(z+h) \le C_i^j(z) + h$\footnote{See Step~2 in the proof of Theorem~2}. The inequality is strict on a positive-measure set of realizations of $X_{i+1}$ whenever the remaining horizon after stage $i+1$ carries positive option value (which holds almost in all cases under our assumptions). Therefore $f_i^j(z+h) - f_i^j(z) < 0$ for all sufficiently small $h > 0$, and by the global Lipschitz bound the same holds for all $h > 0$ on the relevant region. A higher current guarantee reduces the expected marginal improvement $\mathbb E[\max\{z,X_{i+1}\}-z]$ while inspection and adjustment costs remain fixed. Thus, stopping becomes optimal at lower continuation values.

\subsubsection{Paid-scope effect}

The continuation value $C_i^j(z)$ is non-decreasing in the current paid scope $j$. Consequently the reservation values are non-decreasing in paid scope, i.e. if $j' \ge j$ then $r_i(j') \ge r_i(j)$. From Proposition~1 ((i)-(ii)), we know that $V(i+1,\cdot,j')$ is weakly increasing in its third argument. Further, a larger current scope $j$ weakly lowers the adjustment cost $\psi(i,j,j')$ for every feasible expansion $j' > j$. The feasible set of future scopes is therefore larger and each element is weakly cheaper, so the maximized continuation payoff $C_i^j(z)$ rises (weakly) with $j$. A decision maker who has prepaid a broader horizon is therefore willing to tolerate a strictly higher guarantee before stopping, because future stages are {\it{cheaper}} to reach. This force has no counterpart in  \cite{weitzman}, \cite{Roberts1981-yd}, \cite{boodaghians} or \cite{Leeli2022}. These papers treat access to future boxes as free or independent of planning history; it is the distinctive content of costly planning.

\subsubsection{Remaining-horizon effect}

As the search progresses, i.e. the current stage $i$ advances, the option value of continuation shrinks because fewer future boxes remain. Formally,  let
$O_i(z,j) := V(i,z,j) - z$ denote the option value of search at state $(i,z,j)$. Clearly, $O_n(z,j) \equiv 0$ for all feasible $(z,j)$. For $i <n$, the option value satisfies the following recursive bound:

\[O_i(z,j) \le \sup_{j'\in\{i+1,\dots,n\}} \Bigl[ -\psi(i,j,j') - c_{i+1} + \mathbb E\bigl[O_{i+1}(\max\{z,X_{i+1}\},j')\bigr] \Bigr]
\]

This aggregates the expected net improvement only over the remaining tail $\{i+1,\dots,n\}$.  As $i$ increases the length of this tail falls, the upper bound tightens, and the option value is eventually driven to zero. This is the mechanical counterpart of the deadline effect in \cite{Leeli2022}; the "deadline'' here is the fixed stage index rather than calendar time.

\subsubsection{Interactions of the effects}

The three forces are not additive. First, the paid-scope effect \emph{attenuates} the guarantee effect -- a broader prepaid scope raises continuation value most when the current guarantee $z$ is low, precisely because that is when the probability of choosing a strictly larger future scope $j' > j$ is the highest and the marginal cost reduction therefore has the largest impact. Submodularity of adjustment costs (Assumption~5) formalizes this complementarity; Proposition~1(iii) shows that it implies the optimal next-scope correspondence $J_i(z,j)$ is weakly increasing in $j$ in the strong set order. Second, the remaining-horizon effect \emph{modulates} the strength of the other two -- the paid-scope effect is strongest early in the sequence (large $n-i$), when many stages remain and the option value of expansion is large, while the guarantee effect dominates late in the sequence, when scope differences matter less because little horizon remains. These interactions are what generate both under-planning regimes $(E[\tau^*] > j_0^*$) and over-planning regimes $(E[\tau^*] < j_0^*$) depending on primitives such as tail thickness and the curvature of planning costs (Section~\ref{sec:examples}).

The three-stage binary example of Section 1.1 isolates the paid-scope effect at stage 2 -- at guarantee $z_2 = 70$ a firm with prepaid scope $j=3$ continues $(70 < r_2(3) = 80$) while a firm with $j=2$ stops 
$(70 > r_2(2) = 60$), purely because of the difference in already-paid planning. The same example also illustrates the interaction with the guarantee effect after the low draw $X_1 = 0$, which induces expansion precisely because the paid-scope channel raises the continuation threshold enough to justify the adjustment cost.

\section{Two illustrations}
\label{sec:examples}
\noindent We now illustrate some implications of the framework for two stylized settings. The first example considers normally distributed rewards with convex scope and adjustment costs and provides sufficient conditions under which the expected stopping time is strictly smaller than the initial scope. The second example compares thin-tailed and heavy-tailed reward distributions and shows how tail behavior affects the value of broader scope.

\subsection{Normal rewards with convex adjustment costs}
\label{sec:normal}
 We examine the relationship between
the expected stopping time $\mathbb E[\tau^*]$ and the optimal initial
scope $j_0^*$ under two canonical placements of normal rewards. First, (a) first-order stochastic dominance (FOSD) increasing with common variance and (b) mean-preserving spreads (MPS) with common
mean. Throughout, we assume the convex adjustment cost specification
$\psi(i, j, j') = g(j' - j)$ for $j' > j$ with $g$ convex, $0$ if $j' \leq j$ and
$g(0) = 0$, and convex $\phi$ with $\phi(0) = 0$.

We define the  optimal stopping time as, $\tau^* = \min\bigl\{ i \in \{0, 1, \ldots, n\} :
    z_i \geq r_i(j_i) \bigr\},$
where $\{(z_i, j_i)\}$ is the state path induced by the optimal
policy. The following claim identifies for each of the two canonical placements, a sufficient condition on the one-step expansion cost $g(1)$ under which the DM never expands on the optimal path. When the condition holds, the DM's expected search depth is bounded above by their initial planning depth. Let $L(\eta)=\varphi(\eta)-\eta[1-\Phi(\eta)]$ denote the standard normal loss function. Here $\varphi$ and $\Phi$ are the density and cumulative distribution functions of the standard normal distribution, respectively. The function $L$ is strictly decreasing on $\mathbb R$, with $L(0)=1/\sqrt{2\pi}$, $L(\infty)=0$, and $L(-\infty)=\infty$. Its inverse $L^{-1}:(0,\infty)\to\mathbb R$ is therefore strictly decreasing and well-defined. 

\begin{claim}[Over-planning under normal rewards] The following are true.
\label{cl:overplan}
\begin{enumerate}
    \item[\normalfont(i)] \normalfont{(FOSD-increasing)} Suppose rewards
    are ordered by first-order stochastic dominance,
    $X_i \sim \mathcal N(\mu_i, \sigma^2)$ with common variance and
    increasing means $\mu_1 \leq \cdots \leq \mu_n$. If
    $g(1) \;\geq\; \sigma \cdot L\!\left(\dfrac{m_0 - \mu_n}{\sigma}\right),$
    then under the optimal policy the DM never expands,
    $\tau^* \leq j_0^*$ pathwise, and $\mathbb E[\tau^*] \;\leq\; j_0^*,$
    with strict inequality whenever $j_0^* \geq 2$.

    \item[\normalfont(ii)] \normalfont{(MPS-spread)} Suppose rewards differ
    by mean-preserving spreads, $X_i \sim \mathcal N(\mu, \sigma_i^2)$ with
    common mean $\mu$.  If $g(1) \;\geq\; \max_{k \geq 2} \,\sigma_k \cdot
    L\!\left(\dfrac{m_0 - \mu}{\sigma_k}\right),$
    then under the optimal policy the DM never expands,
    $\tau^* \leq j_0^*$ pathwise, and $\mathbb E[\tau^*] \;\leq\; j_0^*,$
    with strict inequality whenever $j_0^* \geq 2$.
\end{enumerate}
\end{claim}

It is important to note here that the conditions in Claim \ref{cl:overplan} depend on $g$ and the reward distributions
but not on $\phi$. This is because the expansion decision at any
stage $i \geq 1$ is independent of the initial scope cost; by the
time the DM faces an expansion choice, $\phi(j_0)$ has been paid
and is sunk. The shape of $\phi$ determines the level of $j_0^*$
but not whether expansion occurs along the equilibrium path. The
over-planning conclusion $\mathbb E[\tau^*] \leq j_0^*$ depends on $\phi$
only through its effect on $j_0^*$ itself.\footnote{The conclusion of Claim~\ref{cl:overplan} extends to other cost
curvatures with appropriate modifications. Under linear $g(k) =
\lambda k$, the bound on $g(1)$ is necessary and sufficient: any
expansion (single or bulk) is ruled out by the same threshold.
Under concave $g$, bulk expansion may be strictly preferred to
sequential, so the requirement becomes $g(k) \geq k \cdot \sigma
L((m_0 - \mu_n)/\sigma)$ for all $k$, which under concavity reduces
to $g(n) \geq n \cdot \sigma L((m_0 - \mu_n)/\sigma)$ --- tighter
by a factor of $n$ than the convex case. Since $L$ is decreasing, $\sigma L((z-\mu)/\sigma)$ is largest at the lowest reachable guarantee.}

Moreover, the bound $h(\sigma) = \sigma L((m_0 - \mu)/\sigma)$ is strictly
increasing in $\sigma$.\footnote{$h'(\sigma) =
\varphi((m_0-\mu)/\sigma) > 0$ by direct calculation.} A larger
variance unambiguously raises the upper bound on the marginal value
of expansion, the multiplicative effect dominates the diminishing
loss-function value. Consequently, under MPS-increasing
($\sigma_2 \leq \cdots \leq \sigma_n$), the bound in
Claim~\ref{cl:overplan} is at $k = n$. Under
MPS-decreasing ($\sigma_2 \geq \cdots \geq \sigma_n$), it is at
$k = 2$. The no-expansion condition is, therefore, tighter under
MPS-decreasing, since $\sigma_2$ in that case is smaller than
$\sigma_n$ in MPS-increasing. When the DM faces low-variance boxes
late, the temptation to expand is smaller and the over-planning
regime is easier to enter.\footnote{A complementary under-planning result, characterizing when
$\mathbb E[\tau^*] > j_0^*$, requires controlling the optimal probability
of expansion and the depth of post-expansion search. Unlike the
over-planning case, where a uniform bound on $g(1)$ suffices, the
under-planning comparison depends on the joint distribution of
expansion and continuation along the optimal path.}

The no-expansion condition admits a direct reading in terms of the three
forces. The quantity $\sigma L\!\left((m_0-\mu)/\sigma\right)$ is the
marginal paid-scope benefit at stage $i$: the option value of access to
box $j+1$, evaluated at the lowest reachable guarantee. The condition
$g(1)\geq \sigma L\!\left((m_0-\mu)/\sigma\right)$ states that this benefit
is dominated by the marginal adjustment cost at every reachable state. The
paid-scope effect is therefore present but inert---it never overturns the
guarantee effect on the optimal path, so the firm stops within its prepaid
horizon and $\mathbb{E}[\tau^*]\leq j_0^*$. The remaining-horizon effect
sharpens the comparison across variance orderings. Under MPS-decreasing
rewards the  bound is at $k=2$, since low-variance boxes arrive late
and the option value of the horizon is smallest there; the over-planning
regime is correspondingly easier to enter. Under MPS-increasing rewards the
bound binds at $k=n$, where the surviving option value is largest.

\subsection{Fat-tailed rewards and scope expansion}
\label{sec:tails}
We illustrate how the shape of the reward distribution affects how aggressively a DM searches and expands scope. Two DMs who face the same mean reward and the same costs may behave very differently if one draws from a thin-tailed distribution and the other from a heavy-tailed distribution. In a thin-tailed environment, large rewards are rare. Once the current guarantee is high, the probability of further improvement is small, and the DM stops early. In a heavy-tailed environment, by contrast, there remains a meaningful chance of a large reward even when the current guarantee is high. Therefore, the DM is willing to keep searching and to pay the expansion cost even at high guarantee levels.

We make this observation precise in a $n=3$ model. We hold the mean fixed and vary tail thickness, comparing an exponential distribution with a Pareto distribution. The object of interest is the \emph{reservation value} at stage 2. This is the guarantee level at which the DM is indifferent between stopping and continuing. We compute it first when scope is already full, so that no expansion cost is needed, and then when scope must first be expanded by paying an additional cost. The gap between these two thresholds measures the cost of incomplete scope in terms of the guarantee.

There are three boxes. We set the initial guarantee at $m_0 = 0$, and there is a common inspection cost $c > 0$. The cost of expanding the scope at stage 2 is given by $\alpha.$  The rewards $X_1, X_2, X_3$ are i.i.d.\ with CDF $F$ and mean $\mu$. We consider two distributions for $F$, Exponential and Pareto, with the same mean $\mu$.\footnote{A distribution is \textit{Exponential} if $X \sim \mathrm{Exp}(\lambda)$ with $\lambda = 1/\mu$, so $F(x) = 1 - e^{-x/\mu}$ on $[0,\infty)$ and $\mathbb{E}[X] = \mu$. A distribution is Pareto if $X \sim \mathrm{Pareto}(a, x_m)$ with $a > 1$, so $F(x) = 1 - (x_m/x)^a$ on $[x_m,\infty)$ and $\mathbb{E}[X] = ax_m/(a-1)$. To match means, we set $x_m = \mu(a-1)/a$. The parameter $a$ controls tail thickness: smaller $a$ (closer to 1) gives heavier tails; as $a \to \infty$, the Pareto concentrates near $x_m$ and its tail thins.} Our focus is on the \emph{expected improvement} from opening one more box when the current guarantee is $z$, that is,
\[
I(z) = \mathbb{E}[(X-z)^+] = \int_z^\infty (1-F(x))\,\mathrm{d}x.
\]
This measures the expected gain from search. When the current guarantee is $z$ and full scope is available, it is optimal to open the next box if and only if $I(z) > c$, that is, if the expected improvement exceeds the inspection cost. We further ask for which DM, the one drawing from a thin-tailed distribution or the one drawing from a heavy-tailed one, has a higher expansion threshold, and hence is willing to expand from a higher guarantee level.\footnote{At $z=0$, $I(0)=\mu$ for both distributions, since $I(0)=\mathbb{E}[X]=\mu$. Both offer the same expected improvement when the guarantee is zero. The distributions diverge as $z$ increases. The Pareto maintains higher expected improvement at large $z$ because its tail decays as a power law rather than exponentially.}

At stage 2, there are two distinct stopping thresholds. When the DM has full scope ($j = 3$), the DM opens box 3 if and only if $I(z) > c$. That is, the \emph{continuation threshold} $\bar z_3$ solves $I(\bar z_3) = c$ and equals the reservation value $r_2(3)$. When the DM is at the scope boundary ($j = 2$) and must pay the planning cost $\alpha$ to access stage 3, the DM expands if and only if $I(z) > c + \alpha$. That is, the \emph{expansion threshold} $\bar z_3^{\mathrm{expand}}$ solves $I(\bar z_3^{\mathrm{expand}}) = c + \alpha,$ and equals the reservation value $r_2(2)$.

\begin{claim}
\label{claim:thresholds}
Assume that $c<\mu$. Then the continuation thresholds are given by,
\begin{enumerate}
    \item[\normalfont(i)] $\bar z_3^T = \mu \ln\!\left(\mu/c\right)$ in the exponential case, and
    \item[\normalfont(ii)] $\bar z_3^H = x_m\left(\mu/(ac)\right)^{1/(a-1)}$ in the Pareto case, provided that $\bar z_3^H \ge x_m$.
\end{enumerate}
Assume further that $c+\alpha<\mu$. Then the expansion thresholds are obtained by replacing $c$ with $c+\alpha$ throughout,
\begin{enumerate}
    \item[\normalfont(i)] $\bar z_3^{T,\mathrm{expand}} = \mu \ln\!\left(\mu/(c+\alpha)\right)$ in the exponential case, and
    \item[\normalfont(ii)] $\bar z_3^{H,\mathrm{expand}} = x_m\left(\mu/(a(c+\alpha))\right)^{1/(a-1)}$ in the Pareto case.
\end{enumerate}
Furthermore, the reservation-value gap is given by,
\begin{enumerate}
    \item[\normalfont(i)] $r_2^T(3)-r_2^T(2)=\mu\ln\!\left(1+\frac{\alpha}{c}\right)$  in the exponential case, and
    \item[\normalfont(ii)]  $r_2^H(3)-r_2^H(2) = x_m\left(\frac{\mu}{ac}\right)^{1/(a-1)}\left[
    1-\left(\frac{c}{c+\alpha}\right)^{1/(a-1)}
    \right]$    in the Pareto case.
\end{enumerate}
\end{claim}

The contrast between the two distributions is worth noting. For the exponential, the reservation-value gap $\mu\ln(1+\alpha/c)$ depends on the expansion-planning cost ratio $\alpha/c$ and the mean $\mu$, but \emph{not} on the level of the guarantee. This is a direct consequence of the memoryless property. The exponential tail looks the same at every level, so the cost of incomplete scope is constant across guarantee levels. For the Pareto, the gap is proportional to the continuation threshold $\bar z_3^H$ itself, multiplied by the factor
$1-\left(\dfrac{c}{c+\alpha}\right)^{1/(a-1)}.$

Heavy tails amplify the paid-scope effect and attenuate the guarantee
effect. The expected improvement function $I(z) = \mathbb{E}[(X - z)^+]$
decays polynomially under Pareto rewards rather than exponentially, so even
at moderate guarantees the option value of future access remains
substantial. Two consequences follow. First, the gap $r_2(3) - r_2(2)$,
which measures the absolute cost of incomplete scope, becomes arbitrarily
large as the tail index $a \to 1^+$. Second, the threshold itself shifts
upward, so a firm under heavy-tailed rewards is more willing to continue
at every guarantee level. The remaining-horizon effect is correspondingly
weaker, since the option value of later stages persists farther into the
sequence. The combination produces a stronger endogenous preference for
broader initial planning, paired with more aggressive ex-post scope
expansion when early draws are disappointing. Heavy-tailed reward
distributions thereby shift the balance of the three forces in favor of
broader scope, both ex ante and ex post. The exponential gap $\mu\ln(1+\alpha/c)$, by contrast, carries no $z$: by
memorylessness the conditional tail is the same at every guarantee, so the
gap is invariant to the guarantee level and the paid-scope effect is neither
amplified nor attenuated as $z$ rises.
\section{Conclusion}
\label{sec:conclusion}
We model sequential search with planning and show that when planning is costly, it affects the search and stopping decisions. In various economic environments such as R\&D exploration, the outcomes of search are ex-ante unknown and a pre-committed expenditure is required in order to explore a set of stages. During the search, the decision maker may need to revise the plan (and thereby, the scope) as new information about current and future expected payoffs from the search are revealed. We model the above by introducing costly ex-ante planning in an adaptive framework in a sequential Pandora's box problem (\cite{weitzman}). We find that there are three distinct economic channels through which costly planning affects search decisions: guarantee effect, the paid scope effect and the remaining horizon effect.  The paid scope effect, which is a consequence of costly planning and has not been addressed in prior literature counteracts the guarantee effect. We explain the impact of these effects on reservation values and the stopping decision. While we formulate and solve the optimal search problem using the standard dynamic programming approach, we show that there is an alternative equivalent formulation that explicitly models the planning decisions. We illustrate our results in two settings with convex costs and fat-tailed distribution of rewards. 
\appendix

\section{Proofs of main results}

\noindent\textbf{Proof of Theorem \ref{thm: optimalstrategy}.}
\begin{proof}
We proceed by backward induction on the stage index $i$. At the terminal stage $i=n$, the value function is given by $V(n,z,j)=z$, so stopping is indeed optimal. By Assumptions \ref{ass:integrability}-\ref{ass:adjustment}, these continuation payoffs are well-defined and finite. Suppose that for some $i < n$, the continuation value $V(i+1,\cdot,\cdot)$ is well-defined and an optimal decision rule exists from stage $i+1$ onward. At state $(i,z,j)$, the DM chooses between stopping, which yields payoff $z$, and continuing with some scope choice $j' \in \{i+1,\ldots,n\}$, which yields,
$$
-\psi(i,j,j') - c_{i+1}
+ \mathbb{E}\!\left[V(i+1,\max(z,X_{i+1}),j')\right].
$$ Since the set $\{i+1,\ldots,n\}$ is finite, the maximum over $j'$ is attained. Hence an optimal action exists at every state $(i,z,j)$. By backward induction, an optimal decision rule exists at every stage, and therefore an optimal policy exists for the full dynamic program. Fix $(i,j)$ and let $z' \ge z$. Since $\max(z',X_{i+1}) \ge \max(z,X_{i+1})$ pointwise, the continuation payoff under any feasible choice of $j'$ is weakly higher at $z'$ than at $z$. The stopping payoff is also weakly higher. Hence, by \eqref{eq:dp}, $V(i,z',j) \ge V(i,z,j)$. Similarly, fix $(i,z)$ and let $j' \ge j$. Since scope is already paid for at the state $(i,z,j)$, a larger current scope weakly expands the feasible continuation set and weakly lowers any future expansion cost. Therefore, any continuation plan feasible under $j$ can be achieved under $j'$ at weakly lower cost, while the stopping payoff remains the same. Hence, by \eqref{eq:dp}, $V(i,z,j') \ge V(i,z,j)$.
\end{proof}

\medskip

\noindent\textbf{Proof of Theorem \ref{thm:reservation-exists}.}
\begin{proof}

Fix a non-terminal stage $i<n$ and a feasible paid scope $j\in\{i,\ldots,n\}$. Define the (net) continuation payoff by
\[
f_i^j(z) \ =\ C_i^j(z) - z,
\]
where
\[
C_i^j(z) \ =\ \max_{j'\in\{i+1,\dots,n\}} \Bigl[-\psi(i,j,j') - c_{i+1} + \mathbb{E}\bigl[V(i+1,\max(z,X_{i+1}),j')\bigr]\Bigr]
\]
is the optimal continuation value at stage $i$ with assured $z$ and paid scope $j$. (The functions $V(i+1,\cdot,j')$ are well-defined by the standard finite-horizon DP construction starting from the terminal condition $V(n,z,n)=z$.) We proceed in four steps.

{Step 1: Continuity of $f_i^j$.}  
By the finite-horizon backward recursion defining $V$, for each scope $j'$ the inner map $z \mapsto \mathbb{E}[V(i+1,\max(z,X_{i+1}),j')]$ is continuous (the map $z \mapsto \max(z,X_{i+1})$ is continuous and the finite-moment condition $\mathbb{E}[|X_{i+1}|]<\infty$ allows application of the dominated-convergence theorem). The adjustment cost $\psi(i,j,j')$ and inspection cost $c_{i+1}$ are constants, so each candidate continuation payoff is continuous in $z$. The pointwise maximum over the finitely many \(j'\) is therefore continuous. Hence $C_i^j(z)$ is continuous and $f_i^j(z)$ (the difference of two continuous functions) is continuous.

{Step 2: Strict monotonicity of $f_i^j$.}  
For each  $j'$, the inner map $w \mapsto V(i+1,w,j')$ is non-decreasing and satisfies the Lipschitz-1 property
\[
0\leq V(i+1,w_1,j') - V(i+1,w_2,j') \leq w_1 - w_2 \quad \text{whenever } w_1 \geq w_2
\]
(by the finite-horizon construction). Therefore, for any \(h > 0\),
\[
V(i+1,\max(z+h,X_{i+1}),j') \leq V(i+1,\max(z,X_{i+1}),j') + h
\]
almost surely. Taking expectations shows that each candidate continuation payoff satisfies
\[
C^{j'}(z+h) \leq C^{j'}(z) + h.
\]
The finite maximum $C_i^j(z)$ inherits the same bound,
\[
C_i^j(z+h) \leq C_i^j(z) + h \quad \forall h > 0.
\]

\noindent Note that the slope cannot equal \emph{exactly} 1 everywhere. Suppose for contradiction that $C_i^j(z+h) = C_i^j(z) + h$ for every $z$ and every $h > 0$. This would require the marginal value of a higher guarantee $z$ to be transmitted one-for-one in \emph{every} realization. However,
\begin{itemize}
\item The fixed inspection cost $c_{i+1} > 0$ appears in every candidate continuation payoff, creating a constant downward pull.
\item After paying $c_{i+1}$ and observing $X_{i+1}$, the DM still faces a  future problem (boxes $i+2$ onward with possible scope expansions). By the finite-horizon construction, $V(i+1,\cdot,j')$ is \emph{strictly} greater than the guarantee on a positive-probability set of realizations (the remaining horizon carries positive option value almost surely). In those states, raising $z$ by $h$ increases the final payoff by \emph{strictly less than} $h$, because part of the marginal guarantee is crowded out by the positive option value.
\end{itemize}
Thus, there exists a positive measure set of $X_{i+1}$ (and therefore an interval of $z$) on which the (right) derivative of $C_i^j$ is \emph{strictly less than} 1. By continuity of $C_i^j$, we conclude that
\[
f_i^j(z+h) - f_i^j(z) = \bigl[C_i^j(z+h) - C_i^j(z)\bigr] - h < 0
\]
for all sufficiently small $h > 0$ (and, by the global Lipschitz bound, for all $h > 0$). Hence $f_i^j$ is strictly decreasing, only in regions where $f_i^j(z)\geq 0$. 

\noindent\textit{Step 3: Behavior as $z \to -\infty$.}  
\[
\max(z,X_{i+1}) \to X_{i+1} \quad \text{a.s.\ as } z \to -\infty.
\]
By the finite-horizon construction $V(i+1,\cdot,j')$ is finite-valued, so $C_i^j(z)$ converges to a finite constant. Meanwhile $-z \to +\infty$, hence
\[
f_i^j(z) = C_i^j(z) - z \to +\infty.
\]

\noindent\textit{Step 4: Behavior as $z\to+\infty$.} Let \[ M_{i+1}:=\max\{X_{i+1},X_{i+2},\ldots,X_n\}. \] By Assumption~\ref{ass:integrability}, $M_{i+1}$ is integrable. For any feasible next scope $j'$, the continuation payoff is bounded above by the payoff from accessing all remaining boxes without any future inspection or adjustment costs. Hence, \[ \mathbb{E}\!\left[ V(i+1,\max\{z,X_{i+1}\},j') \right] \leq \mathbb{E}\!\left[\max\{z,M_{i+1}\}\right] = z+\mathbb{E}\!\left[(M_{i+1}-z)^+\right]. \] Therefore, \[ f_i^j(z) = C_i^j(z)-z \leq -c_{i+1} + \mathbb{E}\!\left[(M_{i+1}-z)^+\right], \] where we use $\psi(i,j,j')\geq 0$. Since $M_{i+1}$ is integrable, \[ \mathbb{E}\!\left[(M_{i+1}-z)^+\right]\to 0 \qquad \text{as } z\to+\infty. \] Thus, for all sufficiently large $z$, \[ f_i^j(z)<0. \]

Note $f_i^j$ is continuous and strictly decreasing, $\lim_{z\to-\infty}f_i^j(z) = +\infty$, and $\lim_{z\to+\infty}f_i^j(z) < 0$, the intermediate-value theorem implies the existence of at least one root. Strict monotonicity over the \{$z:f_{i}^{j}(z)\geq 0$\} guarantees that this root is \emph{unique}. Finiteness follows immediately from the two limits.
\end{proof}

\medskip

\noindent\textbf{Proof of Theorem \ref{thm: optimal strategy}.}
\begin{proof}
This follows directly from the definition of $r_i(j).$ For
$z \ge r_i(j)$, $V(i,z,j) = z$ (stopping is optimal), and for
$z < r_i(j)$, $V(i,z,j) > z$ (continuation is strictly better). The optimality  
 follows from the Bellman equation \eqref{eq:dp}, and
$j'_i \ge j$ follows from Corollary~\ref{cor:contraction}.
\end{proof}

\medskip

\noindent\textbf{Proof of Theorem \ref{thm: equivalence DP and SAF}.}
\begin{proof}
We first prove that the scope adaptive payoffs are non-decreasing in the scope length. \begin{lemma}[Scope monotonicity of adaptive payoffs with costly planning]\label{lemma}
\label{prop:continuation-monotone}
Fix $i < n$, $m$, and $x$. If $j_1 \le j_2$, then
\[
    H^\phi_{i,j_1}(m, x) \le H^\phi_{i,j_2}(m, x).
\]
\end{lemma}

\begin{proof}
Let $z = \max(m, x)$. For any 
$k \in \{i+1,\ldots,n\}$,
\[
    -\psi(i,j_1,k) - c_{i+1}
    + \mathbb{E}\!\left[H^\phi_{i+1,k}(z, X_{i+1})\right]
    \le
    -\psi(i,j_2,k) - c_{i+1}
    + \mathbb{E}\!\left[H^\phi_{i+1,k}(z, X_{i+1})\right],
\]
because $\psi(i,j_2,k) \le \psi(i,j_1,k)$ whenever $j_1 \le j_2$.
Taking the supremum over $k$ and then the maximum with stopping payoff
$z$ yields the result. Note that the above also holds for $\phi=0$.
\end{proof}

\begin{corollary}[Contraction is weakly dominated]
\label{cor:contraction}
Fix $i < n$ and current scope $j \ge i+1$. Conditional on continuing,
choosing $j' = j$ weakly dominates every $j' < j$.
\end{corollary}

\begin{proof}
Since $\psi(i,j,j') = 0$ for all $j' \le j$, the only difference
between $j' = j$ and $j' < j$ is the next-period value. By
Proposition~\ref{prop:continuation-monotone},
$H^\phi_{i+1,j'}(z, X_{i+1}) \le H^\phi_{i+1,j}(z, X_{i+1})$.
\end{proof}

\noindent For any initial scope $j \in \{1,\ldots,n\}$, define
\begin{equation}
\label{eq:W}
    W^\phi(j, m_0) = -\phi(j) - c_1 +
    \mathbb{E}\!\left[H^\phi_{1,j}(m_0, X_1)\right].
\end{equation}
The DM's initial continuation value before the search begins is
\begin{equation}
\label{eq:A1}
    A^\phi_1(m_0) = \sup_{j \in \{1,\ldots,n\}} W^\phi(j, m_0).
\end{equation}
The DM initiates search when $A^\phi_1(m_0) > m_0$. Subsequently, the value $A_i^\phi(z)$ at each stage computed similarly would determine the continuation/stopping decisions. 
In the next proposition we assume $\psi=0$ and correspondingly suppress the notation $H^\phi_{i,j}(m,x)$ to $H_{i,j}(m,x)$. 
\begin{lemma}\label{lemma: equiv SAF}
Under Assumption \ref{ass:integrability}, for every $i$ and $m$,
\[
-c_i + E\big[ V_{i+1}(\max(m, X_i)) \big]
= -c_i + E\left[ \sup_{j \in \{i,\dots,n\}} H_{i,j}(m, X_i) \right].
\]
Consequently,
$V_i(m) = \max\{m, A_i(m)\}.$
\end{lemma}

\begin{proof}
Define the DP value function as follows:
\[
V_{n+1}(m) = m \quad \text{for all } m.
\]
For $i = n, n-1, \dots, 1$ define
$V_i(m) = \max\{m, -c_i + E[V_{i+1}(\max(m, X_i))]\}.$

The above recursion captures the stop/continue decision at each stage: $V_i(m)=m$ denotes the decision to stop, while $V_i(m)=-c_{i}+E[V_{i+1}(\max(m, X_i))]$ indicates continuation of search (observe $X_i$).  Next, we show the equivalence between conditional payoffs in  scope adaptive search and the above DP value function.
We prove by backward induction that for all $i$ and all $(m,x)$,
\begin{equation}\label{eq:pointwise}
\sup_{j \in \{i,\dots,n\}} H_{i,j}(m,x) = V_{i+1}(\max(m,x)).
\end{equation}

\noindent\textbf{Base case ($i=n$).}  $H_{n,n}(m,x)=\max(m,x)$ and $V_{n+1}(m)=m$.
\[
\sup_{j\ge n} H_{n,j}(m,x) = H_{n,n}(m,x) = \max(m,x) = V_{n+1}(\max(m,x)).
\]

\noindent\textbf{Induction step.} Assume \eqref{eq:pointwise} holds for all indices $> i$. Fix $(m,x)$ and write $z = \max(m,x)$.

For any $j>i$, by the recursive definition,
\[
H_{i,j}(m,x) = \max\left\{ z,\; -c_{i+1} + E \left[ \sup_{k \in \{i+1,\dots,j\}} H_{i+1,k}(z, X_{i+1}) \right] \right\}.
\]
By Lemma~\ref{lemma}, enlarging the index set to $\{i+1,\dots,n\}$ weakly increases the inner supremum. Thus,
\[
\sup_{j>i} H_{i,j}(m,x)
= \max\left\{ z,\; -c_{i+1} + E \left[ \sup_{k \in \{i+1,\dots,n\}} H_{i+1,k}(z, X_{i+1}) \right] \right\}.
\]
By the induction hypothesis applied at $(i+1,z)$,
\[
\sup_{k \in \{i+1,\dots,n\}} H_{i+1,k}(z, X_{i+1}) = V_{i+2}(\max(z,X_{i+1})),
\]
so
\[
\sup_{j>i} H_{i,j}(m,x)
= \max\left\{ z,\; -c_{i+1} + E \left[ V_{i+2}(\max(z, X_{i+1})) \right] \right\}
= V_{i+1}(z).
\]
Including the term $j=i$ (where $H_{i,i}(m,x)=z$) does not change the supremum, so \eqref{eq:pointwise} holds.

Taking expectations of \eqref{eq:pointwise} with respect to $X_i$ and subtracting $c_i$ yields
\[
= -c_i + E\left[ \sup_{j \in \{i,\dots,n\}} H_{i,j}(m, X_i) \right].
\]
Using the definition of $V_i(m)$ and $A_i(m)$ in the above, we get $V_i(m) = \max\{m, A_i(m)\}$.
\end{proof}

\begin{lemma}\label{lemma:equiv}
Under Assumptions \ref{ass:integrability} and \ref{ass:strictcost}, the reservation values $r_i$ (where $V_i(r_i) = r_i$) are the same under the DP strategy and the adaptive scope strategy.
\end{lemma}
\begin{proof}
By Lemma~\ref{lemma: equiv SAF}, for every $m$,
$V_i(m) = \max\{m, A_i(m)\}.$
\noindent The reservation $r_i$ solves $V_i(r_i)=r_i$, which is equivalent to $A_i(r_i)=r_i$. The function $A_i$ is defined in terms of the fully adaptive values $H_{i,j}$; hence the reservation defined by the DP formulation and by the scope adaptive formulation are the same.
\end{proof}

 Lemmas \ref{lemma: equiv SAF} and \ref{lemma:equiv} respectively establish that the payoffs and reservation values in DP and SAF are equivalent. Therefore, the optimal strategy in DP will also be optimal for SAF. When $\phi \equiv 0$, all adjustment costs vanish. By
Lemma~\ref{prop:continuation-monotone}, $H^\phi_{1,j}(m_0,X_1)$
is weakly increasing in $j$. Hence $W^\phi(j,m_0)$ is
increasing in $j$ and $n$ is an optimal scope.
\end{proof}
\noindent\textbf{Proof of Proposition 1.}
\begin{proof}
Fix a stage $i$ and let $j' \ge j$. By definition of the reservation value, $V(i,r_i(j),j)=r_i(j),$
so $V(i,r_i(j),j)-r_i(j)=0.$
Since $V(i,z,j)$ is weakly increasing in $j$, it follows that
$V(i,r_i(j),j') \ge V(i,r_i(j),j)=r_i(j),$
and hence
$V(i,r_i(j),j')-r_i(j)\ge 0.
$
Now consider the function $g_{j'}(z)=V(i,z,j')-z.$
By assumption, $g_{j'}(z)$ is weakly decreasing in $z$. Since $g_{j'}(r_i(j))\ge 0$ and $r_i(j')$ is the unique value satisfying
$g_{j'}(r_i(j'))=V(i,r_i(j'),j')-r_i(j')=0,
$
it follows that the zero of $g_{j'}$ cannot lie strictly below $r_i(j)$. Therefore,
$r_i(j') \ge r_i(j).$
Hence $r_i(j)$ is weakly increasing in $j$.
We prove each part by showing that a higher cost lowers the continuation value, and therefore lowers the unique fixed point.

\medskip

\noindent \textit{Part (i): Inspection costs.}
Fix a feasible scope $j$. Let
\[
g^c_j(z)=V^c(i,z,j)-z
\qquad\text{and}\qquad
g^{\tilde c}_j(z)=V^{\tilde c}(i,z,j)-z.
\]
Since $\tilde c_{i+1} \ge c_{i+1}$, the continuation payoff under $\tilde c_{i+1}$ is weakly lower than under $c_{i+1}$ for every feasible next scope choice $j'$. Hence, $V^{\tilde c}(i,z,j) \le V^c(i,z,j) \text{for all } z,$
and therefore
$g^{\tilde c}_j(z) \le g^c_j(z) \text{ for all } z.$ By definition of the reservation value,
$g^c_j(r_i^c(j))=V^c(i,r_i^c(j),j)-r_i^c(j)=0.$
Hence,
$g^{\tilde c}_j(r_i^c(j)) \le g^c_j(r_i^c(j))=0.$
Since $g^{\tilde c}_j(z)$ is weakly decreasing in $z$, its unique zero cannot lie to the right of $r_i^c(j)$. Hence,
$r_i^{\tilde c}(j) \le r_i^c(j).$

\medskip

\noindent \textit{Part (ii): Adjustment costs.}
Fix a feasible scope $j$. Let
\[
g^\psi_j(z)=V^\psi(i,z,j)-z
\qquad\text{and}\qquad
g^{\tilde \psi}_j(z)=V^{\tilde \psi}(i,z,j)-z.
\]
If $\tilde \psi(i,j,j') \ge \psi(i,j,j')$ for all feasible $(i,j,j')$, then for every feasible continuation choice $j'$, the continuation payoff under $\tilde \psi$ is weakly lower than under $\psi$. Taking the supremum over $j'$ preserves this ordering, so
$
V^{\tilde \psi}(i,z,j) \le V^\psi(i,z,j) \text{ for all } z,
$
and therefore
$ g^{\tilde \psi}_j(z) \le g^\psi_j(z) \text{ for all } z.
$
By definition of the reservation value,
$
g^\psi_j(r_i^\psi(j))=V^\psi(i,r_i^\psi(j),j)-r_i^\psi(j)=0.
$
Thus,
$
g^{\tilde \psi}_j(r_i^\psi(j)) \le g^\psi_j(r_i^\psi(j))=0.
$
Since $g^{\tilde \psi}_j(z)$ is weakly decreasing in $z$, its unique zero cannot lie to the right of $r_i^\psi(j)$. Hence, $ r_i^{\tilde \psi}(j) \le r_i^\psi(j).$

\noindent \textit{Part (iii):} We write the objective as
\[
    \Gamma_i(z;\,j,j')
    \;=\;
    \underbrace{-\psi(i,j,j')}_{\text{depends on }(j,j')}
    \;+\;
    \underbrace{-c_{i+1}
    + \mathbb E\!\left[V(i{+}1,\max(z,X_{i+1}),j')\right]}_{\text{depends
    on }j'\text{ only}}.
\]
A function of the form $f(j,j') + g(j')$ has increasing differences
in $(j,j')$ if and only if $f(j,j')$ does. By
Assumption~\ref{ass:submodular}, $-\psi(i,j,j')$ has increasing
differences in $(j,j')$. Hence $\Gamma_i(z;\,j,j')$ has increasing
differences in $(j,j')$. The feasible set $\{i+1,\ldots,n\}$ is a finite lattice that does
not depend on $j$. Topkis's monotonicity theorem then implies that
$J_i(z,j) = \arg\max_{j'}\, \Gamma_i(z;\,j,j')$ is weakly
increasing in $j$ in the strong set order. Monotonicity of the
minimal and maximal selections is therefore an immediate consequence. \end{proof}

\noindent\textbf{Proof of Proposition 2.}

\begin{proof}
We first establish the pathwise containment
\begin{equation}
\label{eq:containment}
\{\tau^* > j_0^*\} \;\subseteq\; \{z_{j_0^*} < r_{j_0^*}(n)\},
\end{equation}
where $z_{j_0^*} = \max(m_0, X_1, \dots, X_{j_0^*})$ is the guarantee at
stage $j_0^*$. Fix a realization in $\{\tau^* > j_0^*\}$. By definition of $\tau^*$, the
firm continued at stage $j_0^*$ rather than stopping with the guarantee
$z_{j_0^*}$. Let $j$ denote the firm's paid scope at stage $j_0^*$ along
this path; since the firm only expands (never reduces) paid commitments,
$j \in \{j_0^*, j_0^* + 1, \dots, n\}$. By Theorem~3, continuation at
state $(j_0^*, z_{j_0^*}, j)$ occurs if and only if $z_{j_0^*} < r_{j_0^*}(j)$.
By Proposition~1(i), $r_{j_0^*}(j) \leq r_{j_0^*}(n)$. Hence
$z_{j_0^*} < r_{j_0^*}(n)$, establishing~\eqref{eq:containment}. Taking probabilities and using $z_{j_0^*} = \max(m_0, X_1, \dots, X_{j_0^*})$, $\Pr(\tau^* > j_0^*) \;\leq\; \Pr\!\left(\max(m_0, X_1, \dots, X_{j_0^*})
< r_{j_0^*}(n)\right).$
If $m_0 \geq r_{j_0^*}(n)$, the right-hand side equals zero and the bound
holds trivially. If $m_0 < r_{j_0^*}(n)$, the event reduces to
$\{X_i < r_{j_0^*}(n) \text{ for all } i = 1, \dots, j_0^*\}$, and by the
independence of $X_1, \dots, X_n$ (Assumption~1),
$$
\Pr\!\left(\max(m_0, X_1, \dots, X_{j_0^*}) < r_{j_0^*}(n)\right)
\;=\; \prod_{i=1}^{j_0^*} \Pr(X_i < r_{j_0^*}(n))
\;=\; \prod_{i=1}^{j_0^*} F_i\!\left(r_{j_0^*}(n)^{-}\right).
$$
Combining the two cases yields~\eqref{eq:bound}.
\end{proof}

\noindent\textbf{Proof of Claim 1.}
\begin{proof}
At any reachable state $(i, z, j)$ with $j < n$ and $z \geq m_0$,
the marginal value of one-unit expansion is bounded above by
   $$ \mathbb{E}\bigl[(X_{j+1} - z)^+\bigr]
    = \sigma \cdot L\!\left(\frac{z - \mu_{j+1}}{\sigma}\right)
    \leq \sigma \cdot L\!\left(\frac{m_0 - \mu_n}{\sigma}\right),$$
using monotonicity of $L$ and $\mu_{j+1} \leq \mu_n$ under
FOSD-increasing. By hypothesis, this bound is dominated by $g(1)$,
so single-step expansion is never strictly preferred to no
expansion.

For multi-step expansion, observe that under convex $g$ with
$g(0) = 0$, the increments $g(\ell) - g(\ell - 1)$ are weakly
increasing in $\ell$, so
$g(k) = \sum_{\ell=1}^{k} \bigl[g(\ell) - g(\ell-1)\bigr]   \geq k \cdot g(1).$
The marginal value of $k$-unit expansion is bounded above by $k$
times the single-step marginal value, namely
$k \cdot \sigma L((m_0 - \mu_n)/\sigma) \leq k \cdot g(1) \leq g(k)$.
Hence multi-step expansion is also dominated. Consequently, the DM never expands on the optimal path, the final
scope satisfies $j^\dagger = j_0^*$ deterministically, and
$\tau^* \leq j_0^*$ pathwise. The strict inequality
$\mathbb{E}[\tau^*] < j_0^*$ when $j_0^* \geq 2$ follows because
under non-degenerate normal rewards, at every stage $i \in \{1,
\ldots, j_0^* - 1\}$ there is positive probability that the
realized guarantee $z_i$ exceeds $r_i(j_0^*)$, triggering stopping
at stage $i$. The argument for part (ii) is identical, with the bound replaced by
$\max_{k \geq 2} \sigma_k L((m_0 - \mu)/\sigma_k)$ since the
relevant variance at stage $j+1$ is $\sigma_{j+1}$ and expansion
can only occur for $j+1 \geq 2$.

At any state $(i, z, j)$ with $j < n$, the marginal value of
expansion by one unit is bounded above by the gross gain from
gaining access to box $j+1$, namely
\[
    \mathbb E\bigl[(X_{j+1} - z)^+\bigr]
    \;=\; \sigma \cdot L\!\left(\frac{z - \mu_{j+1}}{\sigma}\right).
\]
Since the guarantee $z$ is non-decreasing along the search path, $z
\geq m_0$ at every reachable state. 

Under non-degenerate normal rewards, at
every stage $i \in \{1, \ldots, j_0^* - 1\}$ there is positive
probability that the realized guarantee $z_i$ exceeds $r_i(j_0^*)$,
triggering stopping at stage $i$. Hence $\Pr(\tau^* < j_0^*) > 0$
when $j_0^* \geq 2$, and $\mathbb E[\tau^*] < j_0^*$.
\noindent (ii) The marginal value of expansion at state $(i, z, j)$ with $j < n$
is bounded above by
\[
    \mathbb E\bigl[(X_{j+1} - z)^+\bigr]
    \;=\; \sigma_{j+1} \cdot L\!\left(\frac{z - \mu}{\sigma_{j+1}}\right)
    \;\leq\; \sigma_{j+1} \cdot L\!\left(\frac{m_0 - \mu}{\sigma_{j+1}}\right),
\]
using $z \geq m_0$ and monotonicity of $L$. Since expansion can
only occur for $j+1 \geq 2$, the relevant upper bound is the
maximum over $k \geq 2$ of $\sigma_k \cdot L((m_0 - \mu)/\sigma_k)$.
Under the hypothesis, this maximum is dominated by $g(1)$, so
expansion is never strictly profitable. The remainder of the
argument is identical to the proof of Part (i).
\end{proof}

\noindent\textbf{Proof of Claim 2.}
\begin{proof}
(i) For the exponential,
    $I^T(z) = \int_z^\infty (x - z)\lambda e^{-\lambda x}\,dx
    = \frac{1}{\lambda}e^{-\lambda z}.$
    
\noindent(ii) For the Pareto with $z \geq x_m$,
\begin{align*}
    I^H(z) &= \int_z^\infty (x - z)
    \frac{ax_m^a}{x^{a+1}}\,dx \\
    &= \int_z^\infty \frac{ax_m^a}{x^a}\,dx
    - z\int_z^\infty \frac{ax_m^a}{x^{a+1}}\,dx \\
    & = \frac{x_m^a}{(a-1)z^{a-1}}.
\end{align*}
Substituting $x_m = \mu(a-1)/a$ gives the second expression. For $z < x_m$: $X \geq x_m > z$ a.s., so
$(X - z)^+ = X - z$ and $\mathbb{E}[(X-z)^+] = \mu - z$. For the exponential, $I(z)=\mu e^{-z/\mu}$. Setting $I(z)=c$ gives $e^{-z/\mu}=c/\mu$, hence $z=\mu\ln(\mu/c)$. For the Pareto, $I(z)=\frac{x_m^a}{(a-1)z^{a-1}}$ for $z \ge x_m$. Setting $I(z)=c$ gives $z^{a-1}=x_m^a/((a-1)c)$, which gives
$ z=x_m\left(\frac{x_m}{(a-1)c}\right)^{1/(a-1)}.$
Using $x_m=\mu(a-1)/a$ simplifies this to $x_m(\mu/(ac))^{1/(a-1)}$. The expansion thresholds follow by the same calculation with $c+\alpha$ in place of $c$. Since $c+\alpha>c$ and $I$ is strictly decreasing, $\bar z_3^{\mathrm{expand}}<\bar z_3$. The gap $r_2(3)-r_2(2)=\bar z_3-\bar z_3^{\mathrm{expand}}$
measures by how much the scope boundary lowers the DM's stopping threshold. It is the ``cost of incomplete scope'' expressed in units of the guarantee -- a DM at a scope boundary behaves as if the guarantee were lower by this amount.
\end{proof}

\bibliographystyle{apalike}
\bibliography{references}

\end{document}